\documentclass[review,onefignum,onetabnum]{siamart171218}

\usepackage[a4paper,top=2cm,bottom=2cm,left=2.5cm,right=2.5cm,marginparwidth=1.75cm]{geometry}
\usepackage{tikz, circuitikz}
\usetikzlibrary{decorations.pathmorphing, positioning, snakes}
\usetikzlibrary{calc}
\usetikzlibrary{graphs,graphs.standard}
\usepackage{amsmath, amssymb}

\usepackage{ifthen}

\newcommand{\note}[1]{%
  \ifthenelse{\boolean{showtodos}}%
    {\todo[inline]{#1}}%
    {}%
}

\newif\ifstartcompletesineup
\newif\ifendcompletesineup
\pgfkeys{
    /pgf/decoration/.cd,
    start up/.is if=startcompletesineup,
    start up=true,
    start up/.default=true,
    start down/.style={/pgf/decoration/start up=false},
    end up/.is if=endcompletesineup,
    end up=true,
    end up/.default=true,
    end down/.style={/pgf/decoration/end up=false}
}
\pgfdeclaredecoration{complete sines}{initial}
{
    \state{initial}[
        width=+0pt,
        next state=upsine,
        persistent precomputation={
            \ifstartcompletesineup
                \pgfkeys{/pgf/decoration automaton/next state=upsine}
                \ifendcompletesineup
                    \pgfmathsetmacro\matchinglength{
                        0.5*\pgfdecoratedinputsegmentlength / (ceil(0.5* \pgfdecoratedinputsegmentlength / \pgfdecorationsegmentlength) )
                    }
                \else
                    \pgfmathsetmacro\matchinglength{
                        0.5 * \pgfdecoratedinputsegmentlength / (ceil(0.5 * \pgfdecoratedinputsegmentlength / \pgfdecorationsegmentlength ) - 0.499)
                    }
                \fi
            \else
                \pgfkeys{/pgf/decoration automaton/next state=downsine}
                \ifendcompletesineup
                    \pgfmathsetmacro\matchinglength{
                        0.5* \pgfdecoratedinputsegmentlength / (ceil(0.5 * \pgfdecoratedinputsegmentlength / \pgfdecorationsegmentlength ) - 0.4999)
                    }
                \else
                    \pgfmathsetmacro\matchinglength{
                        0.5 * \pgfdecoratedinputsegmentlength / (ceil(0.5 * \pgfdecoratedinputsegmentlength / \pgfdecorationsegmentlength ) )
                    }
                \fi
            \fi
            \setlength{\pgfdecorationsegmentlength}{\matchinglength pt}
        }] {}
    \state{downsine}[width=\pgfdecorationsegmentlength,next state=upsine]{
        \pgfpathsine{\pgfpoint{0.5\pgfdecorationsegmentlength}{0.5\pgfdecorationsegmentamplitude}}
        \pgfpathcosine{\pgfpoint{0.5\pgfdecorationsegmentlength}{-0.5\pgfdecorationsegmentamplitude}}
    }
    \state{upsine}[width=\pgfdecorationsegmentlength,next state=downsine]{
        \pgfpathsine{\pgfpoint{0.5\pgfdecorationsegmentlength}{-0.5\pgfdecorationsegmentamplitude}}
        \pgfpathcosine{\pgfpoint{0.5\pgfdecorationsegmentlength}{0.5\pgfdecorationsegmentamplitude}}
}
    \state{final}{}
}

\usepackage{lipsum}
\usepackage{amsfonts}
\usepackage{graphicx}
\usepackage{epstopdf}
\usepackage{algorithmic}
\ifpdf
  \DeclareGraphicsExtensions{.eps,.pdf,.png,.jpg}
\else
  \DeclareGraphicsExtensions{.eps}
\fi

\newsiamremark{remark}{Remark}
\newsiamremark{hypothesis}{Hypothesis}
\crefname{hypothesis}{Hypothesis}{Hypotheses}
\newsiamthm{claim}{Claim}

\headers{Spanning Trees of SP Graphs up to Automorphism}{Mithra Karamchedu and Lucas Bang}

\title{Generating the Spanning Trees of Series-Parallel Graphs up to Graph Automorphism\thanks{Posted 08/18/2025.\funding{This work was supported in part by National Science Foundation Grant No.\ 1950885, during the summer Mithra Karamchedu participated in a computer science REU program at Harvey Mudd College.}}}

\author{Mithra Karamchedu\thanks{Department of Computer Science and Department of Mathematics, Harvey Mudd College, Claremont, CA 
  (\email{mkaramchedu@hmc.edu}).}
\and Lucas Bang\thanks{Department of Computer Science, Harvey Mudd College, Claremont, CA 
  (\email{bang@cs.hmc.edu}).}
}

\usepackage{amsopn}

\makeatletter
\newcommand*{\addFileDependency}[1]{
  \typeout{(#1)}
  \@addtofilelist{#1}
  \IfFileExists{#1}{}{\typeout{No file #1.}}
}
\makeatother

\newcommand*{\myexternaldocument}[1]{%
    \externaldocument{#1}%
    \addFileDependency{#1.tex}%
    \addFileDependency{#1.aux}%
}

\ifpdf
\hypersetup{
  pdftitle={Generating the Spanning Trees of Series-Parallel Graphs up to Graph Automorphism},
  pdfauthor={Mithra Karamchedu and Lucas Bang}
}
\fi

\myexternaldocument{ex_supplement}

\begin{document}
\nolinenumbers

\maketitle

\begin{abstract}
  In this paper, we investigate the problem of generating the spanning trees of a graph $G$ up to the automorphisms or ``symmetries\rq\rq\ of $G$. After introducing and surveying this problem for general input graphs, we present algorithms that fully solve the case of \emph{series-parallel graphs}, under two standard definitions. We first show how to generate the nonequivalent spanning trees of a \emph{oriented} series-parallel graph $G$ in output-linear time, where both terminals of $G$ have been individually distinguished (i.e.\ applying an automorphism that exchanges the terminals produces a different series-parallel graph). Subsequently, we show how to adapt these oriented algorithms to the case of \emph{semioriented} series-parallel graphs, where we still have a set of two distinguished terminals but neither has been designated as a source or sink. Finally, we discuss the case of \emph{unoriented} series-parallel graphs, where no terminals have been distinguished and present a few observations and open questions relating to them. The algorithms we present generate the nonequivalent spanning trees of $G$ but never explicitly compute the automorphism group of $G$, revealing how the recursive structure of $G$'s automorphism group mirrors that of its spanning trees.
\end{abstract}



\section{Introduction}

In this paper, we examine the algorithmic aspects of generating a graph's spanning trees under the actions of graph automorphism. Specifically, given a input graph $G$, we ask whether we can efficiently generate the nonequivalent spanning trees of $G$ up to $G$'s automorphisms---where two spanning trees $T_1$ and $T_2$ are equivalent when $T_1 = \sigma T_2$ for some $\sigma \in \text{Aut}(G)$. In addition to surveying this problem over general input graphs, we present algorithms that solve the case of series-parallel graphs.

Spanning tree algorithms are of fundamental importance in computer science and discrete mathematics, and numerous combinatorial problems related to counting and generating spanning trees have been extensively studied \cite{chakraborty, 10.5555/1984890, rubey}. In this paper, we build upon existing problems by also considering how the spanning trees of a graph $G$ relate to the symmetric structure of $G$, their ambient graph.

Several researchers have explored questions closely related to those we study here. In their 1990 paper, Kano and Sakamoto prove necessary and sufficient criteria for a labeled, undirected graph $G$ to have a spanning tree fixed by a subgroup of $G$'s automorphisms \cite{DBLP:journals/dm/KanoS90}. In their respective 2008 and 2009 master's theses, Mohr and Van den Boomen study the problem of generating the equivalence classes of $G$'s spanning trees up to \emph{tree isomorphism}---where the spanning trees $T_1$ and $T_2$ are equivalent when there exists an \emph{isomorphism} $\sigma$ such that $T_1 = \sigma T_2$ \cite{mohr, vanDenBoomen}. We take the same problem that Mohr and Van den Boomen examine but place the additional constraint that the isomorphism $\sigma$ is also an automorphism of $G$.

To illustrate this problem, let us consider the diamond graph $D$ below:
\vspace{5pt}
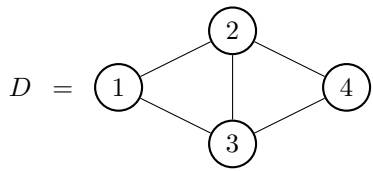
\[
D\hspace{0.5em}=\hspace{0.5em}\begin{tikzpicture}[scale = 0.75, baseline=($(3.base)!.5!(2.base)$)]
\begin{scope}[every node/.style={circle,thick,draw}]
    \node (1) at (0,0) {1};
    \node (2) at (2,1) {2};
    \node (3) at (2,-1) {3};
    \node (4) at (4,0) {4};
\end{scope}

\draw[-] (1) -- (2);
\draw[-] (1) -- (3);
\draw[-] (3) -- (4);
\draw[-] (2) -- (3);
\draw[-] (2) -- (4);
\end{tikzpicture}
\vspace{5pt}\]
We find that $D$ has the following eight spanning trees:
\vspace{10pt}

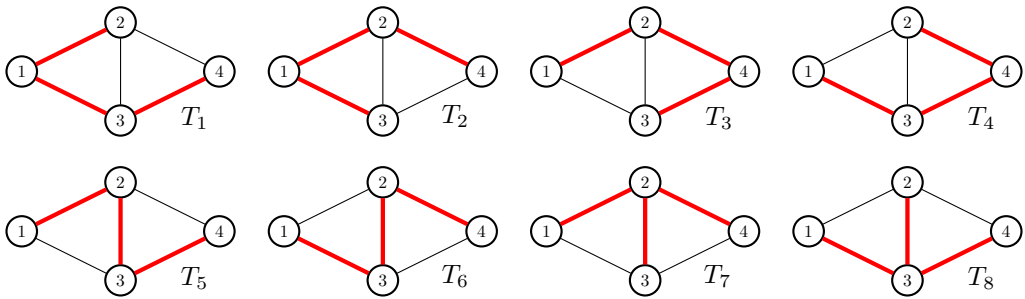
\begin{center}
    \begin{tabular}{cccc}
    \begin{tikzpicture}[scale = 0.65, baseline=($(3.base)!.5!(2.base)$)]
\begin{scope}[every node/.style={scale = 0.65, circle,thick,draw}]
    \node (1) at (0,0) {1};
    \node (2) at (2,1) {2};
    \node (3) at (2,-1) {3};
    \node (4) at (4,0) {4};
\end{scope}

\draw[-, red, ultra thick] (1) -- (2);
\draw[-, red, ultra thick] (1) -- (3);
\draw[-, red, ultra thick] (3) -- (4) node[midway, below right, black] {$T_1$};
\draw[-] (2) -- (3);
\draw[-] (2) -- (4);
\end{tikzpicture} & \begin{tikzpicture}[scale = 0.65, baseline=($(3.base)!.5!(2.base)$)]
\begin{scope}[every node/.style={scale = 0.65, circle,thick,draw}]
    \node (1) at (0,0) {1};
    \node (2) at (2,1) {2};
    \node (3) at (2,-1) {3};
    \node (4) at (4,0) {4};
\end{scope}

\draw[-, red, ultra thick] (1) -- (2);
\draw[-, red, ultra thick] (1) -- (3);
\draw[-] (3) -- (4) node[midway, below right, black] {$T_2$};
\draw[-] (2) -- (3);
\draw[-, red, ultra thick] (2) -- (4);
\end{tikzpicture} & \begin{tikzpicture}[scale = 0.65, baseline=($(3.base)!.5!(2.base)$)]
\begin{scope}[every node/.style={scale = 0.65, circle,thick,draw}]
    \node (1) at (0,0) {1};
    \node (2) at (2,1) {2};
    \node (3) at (2,-1) {3};
    \node (4) at (4,0) {4};
\end{scope}

\draw[-, red, ultra thick] (1) -- (2);
\draw[-] (1) -- (3);
\draw[-, red, ultra thick] (3) -- (4) node[midway, below right, black] {$T_3$};
\draw[-] (2) -- (3);
\draw[-, red, ultra thick] (2) -- (4);
\end{tikzpicture} & \begin{tikzpicture}[scale = 0.65, baseline=($(3.base)!.5!(2.base)$)]
\begin{scope}[every node/.style={scale = 0.65, circle,thick,draw}]
    \node (1) at (0,0) {1};
    \node (2) at (2,1) {2};
    \node (3) at (2,-1) {3};
    \node (4) at (4,0) {4};
\end{scope}

\draw[-] (1) -- (2);
\draw[-, red, ultra thick] (1) -- (3);
\draw[-, red, ultra thick] (3) -- (4) node[midway, below right, black] {$T_4$};;
\draw[-] (2) -- (3);
\draw[-, red, ultra thick] (2) -- (4);
\end{tikzpicture}\\[30pt]
\begin{tikzpicture}[scale = 0.65, baseline=($(3.base)!.5!(2.base)$)]
\begin{scope}[every node/.style={scale = 0.65, circle,thick,draw}]
    \node (1) at (0,0) {1};
    \node (2) at (2,1) {2};
    \node (3) at (2,-1) {3};
    \node (4) at (4,0) {4};
\end{scope}

\draw[-, red, ultra thick] (1) -- (2);
\draw[-] (1) -- (3);
\draw[-, red, ultra thick] (3) -- (4) node[midway, below right, black] {$T_5$};;
\draw[-, red, ultra thick] (2) -- (3);
\draw[-] (2) -- (4);
\end{tikzpicture} & \begin{tikzpicture}[scale = 0.65, baseline=($(3.base)!.5!(2.base)$)]
\begin{scope}[every node/.style={scale = 0.65, circle,thick,draw}]
    \node (1) at (0,0) {1};
    \node (2) at (2,1) {2};
    \node (3) at (2,-1) {3};
    \node (4) at (4,0) {4};
\end{scope}

\draw[-] (1) -- (2);
\draw[-, red, ultra thick] (1) -- (3);
\draw[-] (3) -- (4) node[midway, below right, black] {$T_6$};
\draw[-, red, ultra thick] (2) -- (3);
\draw[-, red, ultra thick] (2) -- (4);
\end{tikzpicture} & \begin{tikzpicture}[scale = 0.65, baseline=($(3.base)!.5!(2.base)$)]
\begin{scope}[every node/.style={scale = 0.65, circle,thick,draw}]
    \node (1) at (0,0) {1};
    \node (2) at (2,1) {2};
    \node (3) at (2,-1) {3};
    \node (4) at (4,0) {4};
\end{scope}

\draw[-, red, ultra thick] (1) -- (2);
\draw[-] (1) -- (3);
\draw[-] (3) -- (4) node[midway, below right, black] {$T_7$};
\draw[-, red, ultra thick] (2) -- (3);
\draw[-, red, ultra thick] (2) -- (4);
\end{tikzpicture} & \begin{tikzpicture}[scale = 0.65, baseline=($(3.base)!.5!(2.base)$)]
\begin{scope}[every node/.style={scale = 0.65, circle,thick,draw}]
    \node (1) at (0,0) {1};
    \node (2) at (2,1) {2};
    \node (3) at (2,-1) {3};
    \node (4) at (4,0) {4};
\end{scope}

\draw[-] (1) -- (2);
\draw[-, red, ultra thick] (1) -- (3);
\draw[-, red, ultra thick] (3) -- (4) node[midway, below right, black] {$T_8$};
\draw[-, red, ultra thick] (2) -- (3);
\draw[-] (2) -- (4);
\end{tikzpicture}
    \end{tabular}
\end{center}
\vspace{10pt}

Furthermore, the automorphism group of $D$ is given by
\begin{align*} \text{Aut}(D) &= \Bigg  \{\, 
\underbrace{\begin{pmatrix}
    1 & 2 & 3 & 4 \\
    1 & 2 & 3 & 4   
\end{pmatrix}}_e,
\underbrace{\begin{pmatrix}
  1 & 2 & 3 & 4 \\
  4 & 2 & 3 & 1   
\end{pmatrix}}_h,
\underbrace{\begin{pmatrix}
  1 & 2 & 3 & 4 \\
  1 & 3 & 2 & 4   
\end{pmatrix}}_v,
\underbrace{\begin{pmatrix}
  1 & 2 & 3 & 4 \\
  4 & 3 & 2 & 1   
\end{pmatrix}}_r
\,\Bigg\},\end{align*}where $e$ is the identity, $h$ is the horizontal flip, $v$ is the vertical flip, and $r = hv$ is the $180^\circ$ rotation. With this, we can generate the orbits of $\textsf{ST}(D)$ over $\text{Aut}(D)$, the sets of mutually equivalent spanning trees. In particular, we have that \[T_1 = v T_2 = r T_3 = h T_4,\qquad T_5 = r T_6, \qquad
T_7 = v T_8,\] so our three orbits of equivalent trees are $\{T_1, T_2, T_3, T_4\}$, $\{T_5, T_6\}$, and $\{T_7, T_8\}$.

Intuitively speaking, the three orbits represent the three fundamentally different ``kinds\rq\rq\ of spanning trees that $D$ has up to its symmetric structure: the first orbit represents paths that circle the outside of $D$, the second represents paths that use the central edge of $D$, and the third represents so-called ``claw\rq\rq\ graphs (with three leaves that hinge on a central vertex).

The problem of generating the spanning trees of a graph $G$ up to automorphism is the same as returning exactly one representative from each orbit. We can therefore say that the nonequivalent spanning trees of $D$ are, for instance, $T_1$, $T_5$, and $T_7$ (among the several ways to choose orbit representatives). Visually, these nonequivalent spanning trees are:
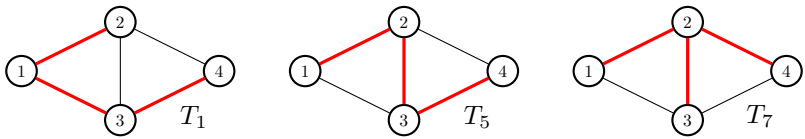
\[\begin{tikzpicture}[scale = 0.65, baseline=($(3.base)!.5!(2.base)$)]
\begin{scope}[every node/.style={scale = 0.65, circle,thick,draw}]
    \node (1) at (0,0) {1};
    \node (2) at (2,1) {2};
    \node (3) at (2,-1) {3};
    \node (4) at (4,0) {4};
\end{scope}

\draw[-, red, very thick] (1) -- (2);
\draw[-, red, very thick] (1) -- (3);
\draw[-, red, very thick] (3) -- (4) node[midway, below right, black] {$T_1$};
\draw[-] (2) -- (3);
\draw[-] (2) -- (4);
\end{tikzpicture}\qquad \begin{tikzpicture}[scale = 0.65, baseline=($(3.base)!.5!(2.base)$)]
\begin{scope}[every node/.style={scale = 0.65, circle,thick,draw}]
    \node (1) at (0,0) {1};
    \node (2) at (2,1) {2};
    \node (3) at (2,-1) {3};
    \node (4) at (4,0) {4};
\end{scope}

\draw[-, red, very thick] (1) -- (2);
\draw[-] (1) -- (3);
\draw[-, red, very thick] (3) -- (4) node[midway, below right, black] {$T_5$};
\draw[-, red, very thick] (2) -- (3);
\draw[-] (2) -- (4);
\end{tikzpicture} \qquad \begin{tikzpicture}[scale = 0.65, baseline=($(3.base)!.5!(2.base)$)]
\begin{scope}[every node/.style={scale = 0.65, circle,thick,draw}]
    \node (1) at (0,0) {1};
    \node (2) at (2,1) {2};
    \node (3) at (2,-1) {3};
    \node (4) at (4,0) {4};
\end{scope}

\draw[-, red, very thick] (1) -- (2);
\draw[-] (1) -- (3);
\draw[-] (3) -- (4) node[midway, below right, black] {$T_7$};
\draw[-, red, very thick] (2) -- (3);
\draw[-, red, very thick] (2) -- (4);
\end{tikzpicture}\]

\subsection{Summary of Paper}
In Section \ref{sec:preliminaries} below, we introduce the prerequisite concepts and notation that will be used throughout the paper. In Section \ref{sec:bruteforce}, we then analyze the complexity of a brute-force algorithm (Algorithm \ref{bruteforce}) that solves this problem, and we briefly discuss special cases in which we can substantially improve upon Algorithm \ref{bruteforce}. In Section \ref{sec:seriesparallel}, we formally define series-parallel graphs and what makes such a graph ``oriented,\rq\rq\ ``semioriented,\rq\rq\ or ``unoriented.\rq\rq\ In Section \ref{sec:oriented}, we fully solve the case of oriented series-parallel graphs and analyze the complexity of the algorithms we propose for them. We then explore semioriented graphs in Section \ref{sec:semioriented}, showing how they largely reduce the oriented case with a few lexicographical modifications. In Section \ref{unoriented}, we make a few observations about unoriented series-parallel graphs, and we highlight the added challenges that they pose over oriented and semioriented graphs. Finally, in Section \ref{sec:conclusion}, we summarize the main results of our paper and offer open questions about generating the nonequivalent spanning trees of series-parallel graphs and beyond.

\section{Preliminaries}
\label{sec:preliminaries}
In this section, we introduce concepts and notation that will be used throughout the paper.

A \emph{spanning tree} $T$ of the graph $G$ is any connected subgraph of $G$ that satisfies two properties: (1) $T$ is a tree (i.e.\ contains no cycles), and (2) $T$ includes or ``spans\rq\rq\,\,every vertex of $G.$ The number of spanning trees of $G$ is denoted $\tau(G)$. For convenience, we also refer to the set of all spanning trees of $G$ as $\textsf{ST}(G).$

A \emph{near tree} $N$ of the graph $G$ is any spanning tree of $G$ minus a single edge. Equivalently, a near tree is a set of $n - 2$ cycle-free edges of $G$ \cite{10.5555/1984890}. We refer to the set of all near trees of $G$ as $\textsf{NT}(G)$.

\indent Intuitively speaking, an \emph{automorphism} is a ``symmetry\rq\rq\,\,of $G$, a way of relabeling $G$'s vertices that preserves adjacency. More formally, let $V(G)$ and $E(G)$ denote the vertex and edge sets of $G$, respectively. An automorphism is an isomorphism from $G$ to itself, a function $\sigma : V(G) \rightarrow V(G)$ such that $\{u, v\} \in E(G)$ if and only if $\{\sigma(u), \sigma(v)\} \in E(G).$  The collection of all such automorphisms of $G$ is known as $G$'s \emph{automorphism group}, written as $\text{Aut}(G)$.

The problem of deciding whether a graph has a nontrivial automorphism group is known as the \emph{graph automorphism problem}, and it is known to be in NP; however, it is unknown whether the problem belongs to P, is NP-complete, or lies in between \cite{lubiw}. 
Among the fastest known algorithms to generate a graph's automorphism group are Darga et al.'s \texttt{saucy}, Junttila and Kaski's \texttt{bliss}, L\'opez-Presa et al.'s \texttt{conauto}, and McKay and Piperno's \texttt{nauty}, which each have strengths in special input cases. None of these algorithms run in worst-case polynomial time but are still quite efficient in practice \cite{lopezpresa}.

For the remainder of the paper, we will assume that $G$ is a simple, labeled, and undirected graph. Given $G,$ we say that the spanning trees $T_1, T_2 \in \textsf{ST}(G)$ are \emph{equivalent} when $T_1 = \sigma T_2$ for some $\sigma \in \text{Aut}(G).$ The \emph{orbit} of any spanning tree $T \in \textsf{ST}(G)$ over $\text{Aut}(G)$ is the set of all spanning trees equivalent to $T$: 
\[\text{Aut}(G)_T = \{\sigma T : \sigma \in \text{Aut}(G)\}.\]
The set of all such orbits is denoted $\textsf{ST}(G)/\text{Aut}(G)$. As mentioned earlier, the problem of generating the nonequivalent spanning trees of $G$ is simply that of finding a representative tree from each orbit. (To reconstruct the orbit of a representative tree, we apply all automorphisms of $G$ to it.)

\section{Naive Brute-Force Analysis}
\label{sec:bruteforce}

If $G$ has a nontrivial automorphism group (i.e.\ contains more than the identity), then the task of determining the orbits of $\textsf{ST}(G)$ over $\text{Aut}(G)$, or even just one representative from each orbit, can be costly. One straightforward brute-force solution is to consider each spanning tree $T$ of $G$ and determine which orbit $T$ belongs to. If $T$ belongs to an orbit we have already discovered, we add $T$ to that orbit; otherwise, we create a new orbit containing $T.$

\indent To check whether $T$ belongs to a given orbit, we assign each orbit $O_i$ for $i \geq 1$ a representative element $R_i$. Then, if there is some automorphism $\sigma \in \text{Aut}(G)$ such that $\sigma T = R_i,$ our spanning tree $T$ must belong in orbit $O_i.$ This gives us the brute-force algorithm below, which returns all orbits---each a set of spanning trees---and their representatives:
\begin{algorithm}[ht]
    \caption{$\texttt{brute\_force}(G)$}
    \label{bruteforce}
    \vspace{3pt}
    \begin{algorithmic}
        \REQUIRE Graph $G$
        \STATE $\text{Aut}(G) \gets \texttt{compute\_automorphism\_group}(G)$
        \STATE $\texttt{orbits} \gets \{\}$
        \FOR{tree $T$ in $\textsf{ST}(G)$}
        \STATE $\texttt{matched} \gets \texttt{false}$
        \FOR{$(R_i, O_i) \in \texttt{orbits}$}
        \IF{$\texttt{not matched}$}
        \FOR{$\sigma \in \text{Aut}(G)$}
        \IF{$\sigma T = R_i$}
        \STATE $O_i \gets O_i \cup \{T\}$
        \STATE $\texttt{matched} \gets \texttt{true}$
        \ENDIF
        \ENDFOR
        \ENDIF
        \ENDFOR
        \IF{$\texttt{not matched}$}
        \STATE $\texttt{orbits} \gets \texttt{orbits} \cup \{(T, \{T\})\}$
        \ENDIF
        \ENDFOR
        \STATE \RETURN \texttt{orbits}
    \end{algorithmic}
    \vspace{3pt}
\end{algorithm}

Although Algorithm \ref{bruteforce} is intuitive, its time complexity can be significant when $G$ has a nontrivial automorphism group. For each spanning tree $T,$ we compare $T$ to a certain number of orbit representatives; for each representative $R_i$, we have to multiply $T$ by $\vert \text{Aut}(G)\vert$ different automorphisms to see if any produce $R_i.$ \footnote{However, for the last representative we compare $T$ to (the ``successful\rq\rq\,\,representative), we may not need to apply all automorphisms before we find a match.} Therefore, the number of automorphism applications our brute-force algorithm performs is given by
\begin{gather*}
    O\left(\sum_{i = 1}^{\tau(G)} k_i \cdot \vert \text{Aut}(G)\vert\right),
\end{gather*}
where $k_i$ is the number of orbits checked for the tree $T_i.$ This in turn simplifies to \[O(\tau(G) \cdot \vert \text{Aut}(G)\vert \cdot \text{[Mean Steps]}),\] where [Mean Steps] is the average over all $k_i$. [Mean Steps] is difficult to quantify precisely, since it depends on several factors, including the distribution of spanning trees across the orbits, the order the orbits are accessed, and the specific sequence in which spanning trees are generated (and how uniform that is relative to the orbits).

According to a 1963 result by Erd\H{o}s and R\'enyi, almost all graphs are \emph{asymmetric}, meaning that they only possess the identity automorphism \cite{erdHos1963asymmetric}. When this is the case, every spanning tree of $G$ belongs to a different orbit and no automorphism multiplications are necessary, so the time complexity of implicitly generating the nonequivalent spanning trees of $G$ is simply $\Theta(\tau(G))$. However, there are still infinitely many graphs $G$ with nontrivial automorphism groups, and these are the graphs for which the problem of generating $\textsf{ST}(G)$'s orbits over $\text{Aut}(G)$ becomes truly relevant.

In the worst case, we apply Algorithm \ref{bruteforce} to a complete graph $K_n$ (with $n$ vertices) without taking the graph's convenient automorphism structure into account ($K_n$ has \emph{every} possible automorphism). In this case, we have that $\tau(G) = n^{n - 2}$ and $\vert \text{Aut}(G) \vert = n!$, which are the largest such values for a graph with $n$ vertices.

For a complete graph, the orbits simply correspond to the different unlabeled trees on $n$ vertices (that is, the different order-$n$ trees up to isomorphism). The number of such trees is given by sequence A000055 in the On-Line Encyclopedia of Integer Sequences (OEIS) \cite{OEIS}; asymptotically, we have that \[\text{A000055}(n) = \Theta\left(\frac{\alpha^n}{n^{2.5}}\right),\]
where $\alpha$ is ``Otter's rooted tree constant,\rq\rq\,\,given by $\alpha = 2.9557652 \ldots$ \cite{finch2003mathematical}\cite{OEIS} Although [Mean Steps] certainly grows with $O(\alpha^n/n^{2.5})$, this bound might not be tight since most of the spanning trees could be clustered in the first few orbits that we check. Substituting our values for $\tau(K_n),$ $\vert \text{Aut}(K_n)\vert,$ and [Mean Steps], the runtime complexity of our brute-force algorithm on $K_n$ is \[O\left(n^{n - 2}\cdot n! \cdot \frac{\alpha^n}{n^{2.5}}\right) = O\big(n^{n - 4.5}\cdot n! \cdot \alpha^n\big),\]where our cost metric is the number of automorphism applications (note that the cost of applying an automorphism scales with the size of our graph, so it too can be significant). This runtime is superexponential in the number of vertices, and it illustrates the sheer complexity of generating $\textsf{ST}(G)$'s orbits over $\text{Aut}(G)$ by resorting to a naive brute-force approach.

However, for many particular cases of graphs, we can significantly improve on this brute-force algorithm by finding ``analytic,\rq\rq\ structure-exploiting approaches tailored to specific graph families. If $G$ is the 1-skeleton of a polyhedron, for instance, then the same problem is equivalent to that of generating the polyhedron's geometrically distinct \emph{unfoldings} \cite{polyhedra, prisms}. If $G$ is the complete bipartite graph $K_{n ,n}$, this problem is the same as generating the unlabeled \emph{equicolorable} trees on $n$ vertices \cite{pippenger}. And, as noted earlier, if $G$ is the complete graph $K_n$, we have the problem of generating the unlabeled trees on $n$ vertices. However, many graphs might be very \emph{close} to a complete graph but lack its full symmetric structure (i.e.\ might equal $K_n$ minus a few edges), meaning that they might not receive as much benefit as a complete graph does over a brute-force approach.

\section{Series-Parallel Graphs}
\label{sec:seriesparallel}
Perhaps the most well-behaved family of graphs with respect to their spanning trees are \emph{series-parallel graphs}, which often represent easy yet conceptually nontrivial cases of graph problems that are difficult in general. In their 1982 paper, Takamizawa, Nishizeki, and Saito show that several NP-complete combinatorial problems have linear-time solutions on series-parallel graphs \cite{DBLP:journals/jacm/TakamizawaNS82}. Building on the work of Smith, Knuth also demonstrates how the problem of generating all spanning trees becomes especially simple on series-parallel graphs due to their constrained recursive structure \cite{10.5555/1984890, smith}.

Intuitively speaking, a series-parallel graph is the underlying graph of a conventional electrical circuit with two terminals. Here, we borrow Knuth's recursive definition from \emph{The Art of Computer Programming}. A graph $G$ is series-parallel when it has two \emph{terminal} vertices, say $s$ and $t,$ and one of the following is true: (1) $G$ is a single edge $s$\,\,--{}--\,\,$t$, (2) $G$ is a \emph{serial superedge} formed by linking together $k \geq 2$ series-parallel subgraphs $G_i$ with terminals $s_i$ and $t_i$, such that $s = s_1,$ $t = t_k,$ and $t_i = s_{i + 1}$ for $1 \leq i < k,$ or (3) $G$ is a \emph{parallel superedge} formed by linking together $k \geq 2$ series-parallel subgraphs $G_i$, such that all $G_i$ have the same terminals $s$ and $t$ \cite{10.5555/1984890}. We call the $G_i$'s the \emph{principal} subgraphs of $G$.

Due to this recursive structure---where serial superedges are built from parallel ones, and vice versa---one especially natural way to represent a series-parallel graph is a \emph{decomposition tree}. If $G$ is a single edge $s$\,\,--{}--\,\,$t$, this tree is a singleton node whose label records the endpoints of $G$. If $G$ is a serial or parallel superedge, the root node of the decomposition tree is an \emph{S-node} or \emph{P-node}, respectively; the child subtrees of our root node are the decomposition trees of each $G_i$. In this decomposition tree, we can enforce that the node types alternate from level to level, so that the children of S-nodes are P-nodes, and the children of P-nodes are S-nodes (we can treat our leaves---the edge nodes---as both S- and P-nodes).
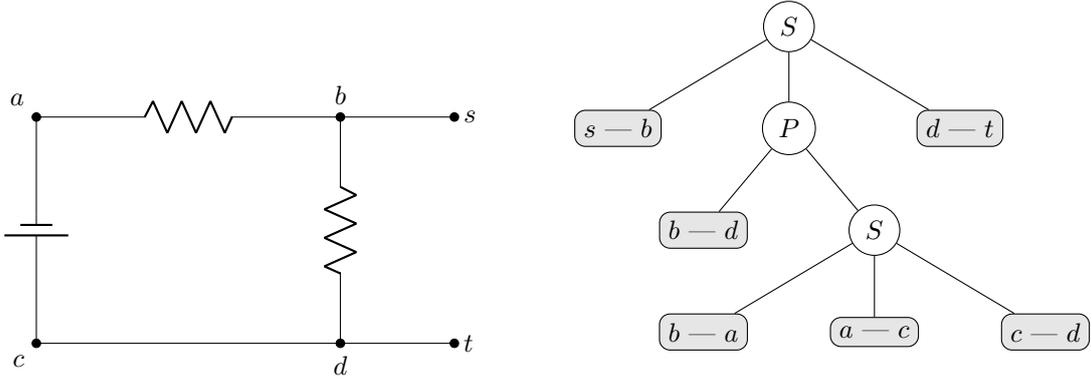
\begin{figure}[h]
    \centering
\begin{circuitikz}[baseline=($(A.base)!.5!(E.base)$)]
		\coordinate (Earth) at (0,-1);
		\coordinate (A) at (0,0);
		\coordinate (B) at (0,3);
		\coordinate (C) at (4,3);
		\coordinate (D) at (4,0);
            \coordinate (E) at (5.5, 3);
            \coordinate (F) at (5.5, 0);
            
		\draw (Earth) (A) node[below left = 0.5mm] {$c$} to[battery1, v>, name=vs, *-] (B) node[above left=0.5mm] {$a$}
				to[R, v^<, name=vr1, *-]  (C) node[above=0.5mm] {$b$}
				to [R, v^<, name=vr2, *-*] (D) node[below=0.5mm] {$d$}
                to (A);

            \draw (Earth) (C) to[short, -*] (E) node[right]{$s$};
            
            \draw (Earth) (D) to[short, -*] (F) node[right]{$t$};
	\end{circuitikz}
    \qquad\quad
    \begin{tikzpicture}[baseline=($(P.center)!.5!(B.center)$), scale = 0.9, draw, edge from parent/.style={draw,-},
    level 1/.style={sibling distance=25mm},
    level 2/.style={sibling distance=25mm},
]
	\node [draw, circle] (S) {$S$}
		child {node [draw, rounded corners, rectangle, fill = gray!20]{$s$ --- $b$}}
            child {
                node [draw, circle] (P) {$P$}
                child { node [draw, rounded corners, rectangle, fill=gray!20] (B1) {$b$ --- $d$} }
                child {node [draw, circle]{$S$}
                         child {node [draw, rounded corners, rectangle, fill=gray!20] {$b$ --- $a$}}
                         child {node [draw, rounded corners, rectangle, fill=gray!20] {$a$ --- $c$}}
                         child {node [draw, rounded corners, rectangle, fill=gray!20] (B) {$c$ --- $d$}}
               } 
            }
            child {node [draw, rounded corners, rectangle, fill=gray!20] {$d$ --- $t$}};
\end{tikzpicture}
    \caption{A voltage divider circuit with two terminals $s$ and $t$ (left) and the decomposition tree of its underlying graph (right).}
\end{figure}

In deciding whether a graph $G$ is series-parallel, however, the set of terminals $\{s, t\}$ of $G$ is not necessarily unique. Furthermore, sources differ on whether $s$ and $t$ are distinguished vertices of $G$, and they often delineate between the concept of a ``series-parallel graph\rq\rq\ and a ``two-terminal series-parallel graph\rq\rq\ accordingly. For our purposes, we will study three natural versions of series-parallel graphs and define equivalence relations on the spanning trees of each, as described below. We call these three versions ``oriented,\rq\rq\ ``semioriented,\rq\rq\  and ``unoriented,\rq\rq\ based on the way the terminals add a natural direction to $G$.
\begin{definition}
    An \emph{oriented} series-parallel graph $(G, s, t)$ consists of a series-parallel graph $G$ (as defined above), a \emph{source} terminal $s$, and a \emph{sink} terminal $t$. We say that $(G, s, t) = (G^\prime, s^\prime, t^\prime)$ when $G = G^\prime$, $s = s^\prime,$ and $t = t^\prime.$
\end{definition}
Note that $s$ is \emph{individually} distinguished as a source and $t$ as a sink: $(G, s, t)$ is a different oriented series-parallel graph from $(G, t, s)$, even though the underlying graph $G$ as well as the set of terminals $\{s, t\}$ are kept the same. As will become important later, an oriented series-parallel graph $(G, s, t)$ is essentially equivalent to a \emph{directed} series-parallel graph, where all edges flow from the source $s$ to the sink $t$ (see Section \ref{isomorphism}).

\begin{definition}A \emph{semioriented} series-parallel graph $(G, \{s, t\})$ consists of a series-parallel graph $G$ and a set of terminals $\{s, t\}$. We say that $(G, \{s, t\}) = (G^\prime, \{s^\prime, t^\prime\})$ when $G = G^\prime$ and $\{s, t\} = \{s^\prime, t^\prime\}$.
\end{definition}
Although the graph still has two terminals $s$ and $t$, neither is designated as a source or sink. Therefore, the semioriented graph $(G, \{s, t\})$ remains the same even when we apply an automorphism that exchanges $s$ and $t$.

In an \emph{unoriented} series-parallel graph $G$, neither of the terminals has been distinguished (i.e.\ $G$ is a vanilla series-parallel graph). That is, $G$ is series-parallel, but there may be several ways of choosing two terminals $s$ and $t$ among $G$'s vertices to form an oriented series-parallel graph. Since unoriented series-parallel graphs have the least amount of imposed structure of the three variants, they prove to be the most challenging case (see Section \ref{unoriented}).

With these definitions in mind, we can now explore the problem of generating the nonequivalent spanning trees of $G$, where $G$ is oriented, semioriented, or unoriented. Let us begin with the oriented case.

\section{Oriented Graphs}
\label{sec:oriented}
In order to generate the nonequivalent spanning trees of oriented series-parallel graphs, we will begin by proving a few helpful results about the graphs' automorphism structure. First, we will define what it means for two spanning trees to be \emph{equivalent} under this said structure, accounting for the fact that our source and sink are each distinguished. Let $G$ be an oriented series-parallel graph with source $s$ and sink $t$, and let $\text{Aut}_\text{or}(G)$ denote the subgroup of $G$'s automorphisms that fix $s$ and $t$:
\begin{definition}
    Given an oriented series-parallel graph $(G, s, t)$, we write that $\emph{Aut}_\emph{or}(G, s, t) = \{\sigma \in \emph{Aut}(G) : \sigma(s) = s\text{ and }\sigma(t) = t\}.$ For brevity, we refer to $\emph{Aut}_\emph{or}(G, s, t)$ as simply $\emph{Aut}_\emph{or}(G)$.
\end{definition}
We then define our equivalence relation on the spanning trees as follows, accounting for the fact that the terminals are designated.
\begin{definition}
  Let $T_1$ and $T_2$ be spanning trees of the oriented series-parallel graph $G$. Then, $T_1$ and $T_2$ are equivalent if and only if $T_1 = \sigma T_2$ for some $\sigma \in \emph{Aut}_\emph{or}(G).$
\end{definition}
Next, we prove results about the structure of $\text{Aut}_\text{or}(G)$ when $G$ is a serial or parallel superedge. In particular, we show that $G$'s automorphisms are built recursively, and this recursive structure mirrors that of $G$'s spanning trees.

\subsection{Serial Superedges} If $G$ is a serial superedge, then the spanning trees of $G$ are simply the union of the spanning trees of each $G_i$. That is, we can form the spanning trees of $G$ by considering every way to choose a spanning tree on each $G_i$; as a result, $\vert \textsf{ST}(G)\vert = \prod_{i = 1}^k \vert \textsf{ST}(G_i)\vert$. The oriented automorphisms of a serial superedge $G$ have exactly the same recursive structure: the oriented automorphisms of $G$ are the union of the oriented automorphisms of each $G_i$, so $\text{Aut}_\text{or}(G) \cong \prod_{i = 1}^k \text{Aut}_\text{or}(G_i)$. We prove the latter result below.

Since the spanning trees and automorphisms of oriented serial superedges $G$ share this structure, we might expect that the spanning trees \emph{up to automorphism} also do. In Theorem \ref{serialtrees}, we show that this is indeed the case.

\begin{theorem}
    \label{serialaut}
    Let $G$ be an oriented serial superedge with terminals $s = s_1, s_2, \ldots, s_{k + 1} = t$ such that each $G_i$ is non-serial. Then, $\emph{Aut}_\emph{or}(G) \cong \prod_{i = 1}^k \emph{Aut}_\emph{or}(G_i).$ 
\end{theorem}
\begin{proof}
    The terminals $s_1, s_2, \ldots, s_{k + 1}$ are the unique points of $G$ that are included in every path $P$ from $s_1$ to $s_{k + 1}$. Moreover, $P$ always visits these vertices in this order: $s_1, s_2, \ldots, s_{k + 1}$.
    
    Now, let $\sigma \in \text{Aut}_\text{or}(G).$ Due to the bijective properties of automorphisms, $\sigma(s_1),$ $\sigma(s_2),$ \ldots, $\sigma(s_{k + 1})$ are the unique vertices of $G$ included in every path $P^\prime$ from $\sigma(s_1)$ to $\sigma(s_{k + 1})$, and $P^\prime$ visits these vertices in this order. Since $\sigma$ preserves the path order of our terminal vertices, $\sigma(s_i)$ must be the $i$th terminal along the path from $\sigma(s_1)$ to $\sigma(s_{k + 1})$. But $\sigma(s_1) = s_1$ and $\sigma(s_{k + 1}) = s_{k + 1}$, so $\sigma(s_i)$ is simply $s_i$. Hence, we have that $\sigma(s_i) = s_i$ for $1 \leq i \leq k$.
    
    We observe that $v$ belongs to subgraph $G_i$ if and only if there exists a path from $s_i$ to $t_i = s_{i + 1}$ containing $v$.
    It is not hard to see, then, that $G_i = \sigma G_i$ for all $\sigma \in \text{Aut}_\text{or}(G)$. It is a matter of formality to show that the function $f : \text{Aut}_\text{or}(G) \rightarrow \prod_{i = 1}^k \text{Aut}_\text{or}(G_i)$ given by \[f(\sigma) = (\sigma_1, \sigma_2, \ldots, \sigma_k),\] where $\sigma_i$ is the restriction of $\sigma$ to $G_i$, is an isomorphism between $\text{Aut}_\text{or}(G)$ and $\prod_{i = 1}^k \text{Aut}_\text{or}(G_i)$. Thus, we conclude that $\text{Aut}_\text{or}(G) \cong \prod_{i = 1}^k \text{Aut}_\text{or}(G_i)$, as desired.
\end{proof}
Since the spanning trees and automorphisms of oriented serial superedges share this structure, we might expect that the spanning trees up to automorphism do as well. In the following theorem, we show that this is indeed the case.

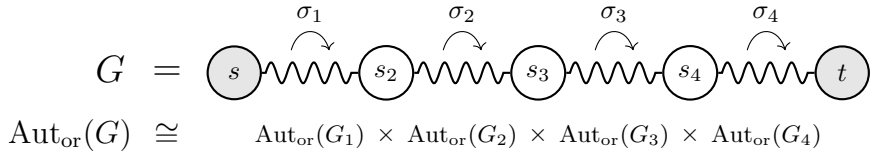
\begin{figure}[h]
\centering

\begin{tikzpicture}[scale=1, every node/.style={scale=1}]

\node[anchor=west] (G) at (-1.95, 0.05) {\Large $G$};
\node[anchor=west] (Equals) at (-1.2, 0) {\Large $=$};
\node[anchor=west] (cong) at (-1.175, -0.82) {\large $\cong$};
\node[anchor=west] (AutOrG) at (-3.1, -0.82) {\large $\text{Aut}_\text{or}(G)$};
\node[circle,thick,draw,fill=gray!20, minimum size=0.7cm] (v1) at (0, 0) {$s$};
\node[circle, thick,draw] (v2) at (2, 0) {$s_2$};
\node[below = 0.22cm of v2] {\small $\times$};
\node[circle,thick,draw] (v3) at (4, 0) {$s_3$};
\node[below = 0.22cm of v3] {\small $\times$};
\node[circle,thick,draw] (v4) at (6, 0) {$s_4$};
\node[below = 0.22cm of v4] {\small $\times$};
\node[circle,thick,draw,fill=gray!20,minimum size=0.7cm] (v5) at (8, 0) {$t$};

\draw[decorate, decoration={complete sines, segment length=0.15cm,  start down, end up, amplitude=2.75mm, pre length = 0.5mm, post length = 0.5mm}, thick] (v1) -- (v2) node[midway, below=0.55cm] {\small $\text{Aut}_\text{or}(G_1)$};
\draw[decorate, decoration={complete sines, segment length=0.15cm,  start down, end up, amplitude=2.75mm, pre length = 0.5mm, post length = 0.5mm}, thick] (v2) -- (v3) node[midway, below=0.55cm] {\small $\text{Aut}_\text{or}(G_2)$};
\draw[decorate, decoration={complete sines, segment length=0.15cm,  start down, end up, amplitude=2.75mm, pre length = 0.5mm, post length = 0.5mm}, thick] (v3) -- (v4) node[midway, below=0.55cm] {\small $\text{Aut}_\text{or}(G_3)$};
\draw[decorate, decoration={complete sines, segment length=0.15cm,  start down, end up, amplitude=2.75mm, pre length = 0.5mm, post length = 0.5mm}, thick] (v4) -- (v5) node[midway, below=0.55cm] {\small $\text{Aut}_\text{or}(G_4)$};

\draw[->] (0.75, 0.3) to[bend left, out=75,in=105,looseness=1.75] node[above] {$\sigma_1$} (1.25, 0.3);
\draw[->] (2.75, 0.3) to[bend left, out=75,in=105,looseness=1.75] node[above] {$\sigma_2$} (3.25, 0.3);
\draw[->] (4.75, 0.3) to[bend left, out=75,in=105,looseness=1.75] node[above] {$\sigma_3$} (5.25, 0.3);
\draw[->] (6.75, 0.3) to[bend left, out=75,in=105,looseness=1.75] node[above] {$\sigma_4$} (7.25, 0.3);
\end{tikzpicture}
\caption{The oriented automorphism group $\emph{Aut}_\emph{or}(G)$ of a serial superedge $G$, with non-serial principal subgraphs $G_1$, $G_2$, $G_3$, and $G_4$ (see Theorem \ref{serialaut}). Every automorphism $\sigma \in \emph{Aut}_\emph{or}(G)$ is the product of automorphisms $\sigma_i \in \emph{Aut}_\emph{or}(G_i)$.}
\end{figure}
\begin{theorem}
    \label{serialtrees}
    Let $G$ be an oriented serial superedge with terminals $s$ and $t$, such that each $G_i$ is non-serial.  Then, the nonequivalent spanning trees of $G$ are formed from the union of the nonequivalent spanning trees of each $G_i$. 
\end{theorem}
\begin{proof}
    Every spanning tree of $G$ is the union of spanning trees on each $G_i$, and by Theorem \ref{serialaut}, each automorphism of $G$ induces an automorphism on each $G_i$. Hence, two spanning trees $T = (T_1, T_2, \ldots, T_k)$ and $T^\prime = (T_1^\prime, T_2^\prime, \ldots, T_k^\prime)$ are equivalent if and only if $T_i = \sigma_i T_i^\prime$ for $1 \leq i \leq k,$ where $\sigma_i \in \text{Aut}_\text{or}(G_i)$. To form the nonequivalent spanning trees of $G$, we therefore only need to compute the nonequivalent spanning trees on each $G_i$ (under $\text{Aut}_\text{or}(G_i)$) are take their union.
\end{proof}
For reasons that emerge later (see Section \ref{sec:parallelsuperedges}), it is also important for us to generate the nonequivalent \emph{near} trees of serial superedges. If $G$ is a serial superedge, then its vanilla near trees (without automorphism reduction) are formed by choosing spanning trees on all but one principal subgraph, on which we choose a near tree instead. As we see below, the nonequivalent near trees up to $\text{Aut}_\text{or}(G)$ are formed the same way, by taking the union of the nonequivalent near trees of one $G_j$ with the nonequivalent spanning trees of the others (across all $k$ choices of $G_j$).
\begin{theorem}
    \label{serialneartrees}
    Let $G$ be an oriented serial superedge with terminals $s$ and $t$, such that each $G_i$ is non-parallel. Then, the nonequivalent near trees of $G$ are the union of the nonequivalent near trees of some $G_j$ with the nonequivalent spanning trees of all the other $G_i$, for $1 \leq j \leq k$. 
\end{theorem}
\begin{proof}
    To show this result, we can follow similar reasoning to Theorem \ref{serialtrees}, accounting for the fact that one of our principal subgraphs has a near tree rather than a spanning tree.
\end{proof}
Based on these theorems, we can introduce the following algorithms to generate (1) the nonequivalent spanning trees and (2) \emph{both} the nonequivalent spanning and near trees of a serial superedge $G$.
\begin{algorithm}[h]
    \caption{$\texttt{serial\_spanning\_trees}(G)$}
    \label{alg:serialspanning}
    \begin{algorithmic}
         \REQUIRE{Serial superedge $G$, with subgraphs $G_1, G_2, \ldots, G_k$.}
         \REQUIRE{Each $G_i$ is non-serial, for $1 \leq i \leq k$.}\newline
         
         \IF{$G$ is an edge} \RETURN{$G$} \hfill\COMMENT{Base case.}
         \ENDIF\newline
         
         \STATE{$\texttt{noneq\_trees} \gets \prod_{i = 1}^k \texttt{parallel\_spanning\_trees}(G_i)$}
         \hfill\COMMENT{Compose all choices of nonequivalent trees.}
         \RETURN{\texttt{noneq\_trees}}
    \end{algorithmic}
\end{algorithm}

\begin{algorithm}[h]
    \caption{$\texttt{serial\_both\_trees}(G)$}
    \label{alg:serialboth}
    \begin{algorithmic}
         \REQUIRE{Serial superedge $G$, with subgraphs $G_1, G_2, \ldots, G_k$.}
         \REQUIRE{Each $G_i$ is non-serial, for $1 \leq i \leq k$.}\newline
         
         \IF{$G$ is an edge} \RETURN{$G$, $\overline{G}$} \hfill\COMMENT{Base case.}
         \ENDIF\newline
         
         \STATE $\texttt{spanning\_trees} \gets \{\}$
         \STATE $\texttt{near\_trees} \gets \{\}$\newline
         \FOR{$1 \leq i \leq k$}
         
            \STATE $\texttt{both\_trees} \gets \texttt{parallel\_spanning\_and\_near\_trees}(G_i)$
            
            \STATE $\texttt{spanning\_trees} \gets \texttt{both\_trees[0]}$ 
            \STATE $\texttt{near\_trees} \gets \texttt{both\_trees[1]}$ 
         \ENDFOR\newline

         \STATE $\texttt{noneq\_spanning\_trees} \gets \texttt{spanning\_trees[0]}$\newline
         
         \FOR{$2 \leq i \leq k$}
            \STATE $\texttt{noneq\_spanning\_trees} \gets \texttt{compose}(\texttt{noneq\_spanning\_trees}, \texttt{spanning\_trees[i]})$ 
         \ENDFOR\newline

         \STATE{$\texttt{noneq\_near\_trees} \gets \{\}$}\newline
         
         \FOR[Generate nonequivalent near trees.]{$1 \leq i \leq k$}
         
            \STATE $\texttt{current} \gets \texttt{near\_trees[i]}$\newline
            
            \FOR{$1 \leq j \leq k$, $j \neq i$}
                \STATE $\texttt{current} \gets \texttt{compose}(\texttt{current},\ \texttt{spanning\_trees[j]})$
            \ENDFOR\newline
            
            \STATE $\texttt{noneq\_near\_trees} \gets \texttt{noneq\_near\_trees} \cup \texttt{current}$
         \ENDFOR\newline
         
         \RETURN{\texttt{noneq\_spanning\_trees}, \texttt{noneq\_near\_trees}}
    \end{algorithmic}
\end{algorithm}

\subsection{Parallel Superedges}
\label{sec:parallelsuperedges}
Parallel superedges have an appealing duality with serial superedges, where the roles of the spanning trees and near trees have been reversed. In particular, the near trees of a parallel superedge $G$ are formed by taking a near tree on each $G_i$, and the spanning trees are formed by taking a spanning tree on one such $G_i$ and near trees on the remaining $G_j$ \cite{10.5555/1984890}.

However, parallel superedges have a more nuanced automorphism structure than serial superedges. In a sense, Theorem \ref{serialaut} tells us that serial superedges $G$ are ``reductionist\rq\rq: the oriented automorphisms of $G$ are exactly equal to the product of the oriented automorphisms of each $G_i$. Parallel superedges, on the other hand, also have ``emergent\rq\rq\ oriented automorphisms that let us \emph{swap} isomorphic $G_i$, as we see below.

\begin{definition}
    Let $G$ be an oriented parallel superedge, such that each $G_i$ is non-parallel. If $G_a$ is a principal subgraph of $G$, then we call an isomorphism $\tau$ of $G_a$ \emph{oriented} if $\tau(s) = s$ and $\tau(t) = t$. 
\end{definition}

\begin{definition}
    Let $G$ be an oriented parallel superedge, such that each $G_i$ is non-parallel, and let $G_a$ and $G_b$ be two principal subgraphs of $G$. If $G_a = \tau G_b$ for some oriented isomorphism $\tau$,  we say that $G_a$ and $G_b$ are \emph{oriented-isomorphic}. For brevity, we write that $G_a \equiv_\emph{or} G_b$.
\end{definition}

\begin{theorem}
    \label{parallelaut}
    Let $G$ be an oriented parallel superedge with terminals $s$ and $t$, such that each $G_i$ is non-parallel. Then, $\emph{Aut}_\emph{or}(G) \cong \emph{swaps}(G) \times \prod_{i = 1}^k \emph{Aut}_\emph{or}(G_i)$, where $\emph{swaps}(G) = \{\sigma \in S_n : G_i \equiv_\emph{or} G_{\sigma(i)}\emph{ for $1 \leq i \leq k$}\}.$
\end{theorem}
\begin{proof}
     We observe that the vertices $u, v \in V(G)$ belong to the same principal subgraph $G_i$ if and only if there exists a path from $u$ to $v$ that does not contain $s$ or $t$ as interior nodes. Now, let $\sigma \in \text{Aut}_\text{or}(G)$, and let $u, v \in V(G_i)$ for some $G_i$. Due to the bijective properties of automorphisms, there exists a path from $\sigma(u)$ to $\sigma(v)$ that does not contain $\sigma(s) = s$ or $\sigma(t) = t$ as interior nodes, so $\sigma(u)$ and $\sigma(v)$ belong to the same principal subgraph $G_j$. The converse is also true: if $\sigma(u), \sigma(v)$ belong to the same $G_j$, then $u$ and $v$ belong to the same $G_i$.
     
     Therefore, $\sigma G_i = G_j$, where $G_j$ is another principal subgraph of $G$. It is clearly true that $G_i$ is oriented-isomorphic to $\sigma G_i$; we have just showed that $\sigma G_a$ is  another principal subgraph of $G$ (i.e.\ $\sigma G_i$ is not spread between different principal subgraphs of $G$).
     
     As a result, $\sigma$ ``swaps\rq\rq\ principal subgraphs of $G$ that are isomorphic up to oriented isomorphisms. Hence, every $\sigma \in \text{Aut}_\text{or}(G)$ is of the form $s\tau$, where $s \in \text{swaps}(G)$, $\sigma_i \in \text{Aut}_\text{or}(G)$ for $1 \leq i \leq k$, and  $\tau  \in \prod_{i = 1}^k \text{Aut}_\text{or}(G_i)$ is given by $\tau = \sigma_1\sigma_2\cdots \sigma_k$. Therefore, $\text{Aut}_\text{or}(G)$ is isomorphic to a subgraph of $\text{swaps}(G) \times \prod_{i = 1}^k \text{Aut}_\text{or}(G_i)$. The other direction, that every automorphism of the form $s\tau$ is an automorphism of $G$, is straightforward. This proves the desired result.
\end{proof}
\begin{figure}[h]
\centering

\begin{tikzpicture}[scale=1, every node/.style={scale=1}]

\node[anchor=west] (G) at (-1.95, 0.05) {\Large $G$};
\node[anchor=west] (Equals) at (-1.2, 0) {\Large $=$};
\node[circle,thick,draw,minimum size=0.7cm,fill=gray!20] (v1) at (0, 0) {$s$};
\node[circle,thick,draw,minimum size=0.7cm,fill=gray!20] (v5) at (7, 0) {$t$};

\draw[bend left=50,-, thick, black,decorate,decoration={complete sines,start down, end up, aspect=0,amplitude=2mm, segment
length=2.5mm, pre length = 1mm}]  (v1) to node [above = 0.2cm] {\footnotesize $\text{Aut}_\text{or}(G_1)$}(v5);

\draw[->] (1.8, 1.8) to[bend left, out=75,in=105,looseness=1.75] node[above=-0.5mm] {\scriptsize $\sigma_1$} (2, 1.8);

\draw[->] (1.8, 0.9) to[bend left, out=75,in=105,looseness=1.75] node[above=-0.5mm] {\scriptsize $\sigma_2$} (2, 0.9);

\draw[->] (1.8, 0.15) to[bend left, out=75,in=105,looseness=1.75] node[above=-0.5mm] {\scriptsize $\sigma_3$} (2, 0.15);

\draw[->] (1.8, -0.9) to[bend right, out=-75,in=-105,looseness=1.75] node[below=-0.5mm] {\scriptsize $\sigma_4$} (2, -0.9);

\draw[->] (1.8, -1.75) to[bend right, out=-75,in=-105,looseness=1.75] node[below=-0.5mm] {\scriptsize $\sigma_5$} (2, -1.75);

\draw[bend left=23,-, thick, black,decorate,decoration={complete sines,start down, end up, aspect=0,amplitude=2mm, segment
length=2.5mm, pre length = 1mm}]  (v1) to node [above = 0.15cm] {\footnotesize $\text{Aut}_\text{or}(G_2)$} (v5);
\draw[bend left=0,-, thick, black,decorate,decoration={complete sines,start down, end up, aspect=0,amplitude=2mm, segment
length=2.5mm, pre length = 1mm}]  (v1) to node [above = 0.1cm] {\footnotesize $\text{Aut}_\text{or}(G_3)$} (v5);
\draw[bend left=-23,-, thick, black,decorate,decoration={complete sines,start down, end up, aspect=0,amplitude=2mm, segment
length=1mm, pre length = 1mm}]  (v1) to node [below = 0.1cm] {\footnotesize $\text{Aut}_\text{or}(G_4)$} (v5);
\draw[bend left=-50,-, thick, black,decorate,decoration={complete sines,start down, end up, aspect=0,amplitude=2mm, segment
length=1mm, pre length = 1mm}]  (v1) to node [below = 0.15cm] {\footnotesize $\text{Aut}_\text{or}(G_5)$} (v5);

\draw[draw=black, fill=gray, fill opacity=0.075] (1.5,-0.25) rectangle ++(4,2.75);
\draw[draw=black, fill=gray, fill opacity=0.15] (1.5,-0.5) rectangle ++(4,-2.0);
\node[anchor=west] (E1) at (5.6,2.1) {\Large $E_1$};
\node[anchor=west] (E2) at (5.6,-2.1) {\Large $E_2$};
\end{tikzpicture}
\caption{The oriented automorphism group $\emph{Aut}_\emph{or}(G)$ of a parallel superedge $G$, with non-parallel principal subgraphs $G_1$, $G_2$, $G_3$, $G_4$, and $G_5$ (see Theorem \ref{parallelaut}). The subgraphs $G_1$, $G_2$, and $G_3$ belong to the same oriented isomorphism class $E_1$, whereas $G_4$ and $G_5$ belong to $E_2$. Every automorphism $\sigma \in \emph{Aut}_\emph{or}(G)$ is the product of automorphisms $\sigma_i \in \emph{Aut}_\emph{or}(G_i)$ and swaps $s \in \emph{swaps}(G)$ between subgraphs belonging to the same class.}
\end{figure}
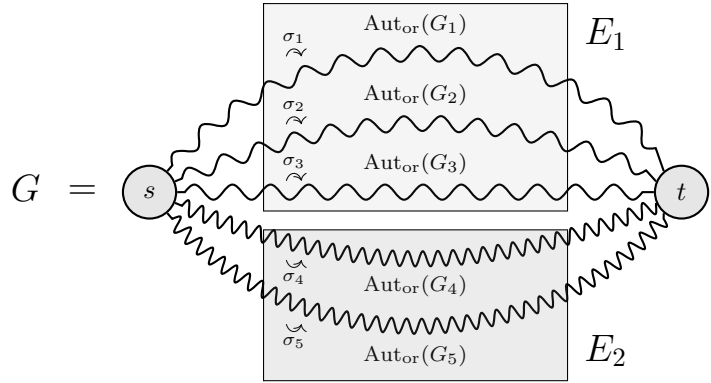
\subsubsection{Near Trees on Parallel Superedges} In essence, Theorem \ref{parallelaut} says that we can swap the vertices of any two principal subgraphs (and preserve $G$) if and only if they are oriented-isomorphic. Therefore, if we partition the $G_i$ into oriented isomorphism classes $E_1, E_2, \ldots, E_\ell$, any swap operations that make one spanning/near tree equal to another must occur between subgraphs in the same $E_i$. We formally define these classes below.
\begin{definition}
    Let $G$ be an oriented parallel superedge, such that each $G_i$ is non-parallel. Then, let us partition $\textsf{\emph{ST}}(G)$ into $\ell$ \emph{oriented isomorphism classes} $E_1, E_2, \ldots, E_\ell$ such that $G_a$ and $G_b$ are in the same $E_i$ if and only if $G_a \equiv_\emph{or} G_b$. For each $E_i$, let the subgraph $R_i \in E_i$ be its representative element.
\end{definition}
With Theorem \ref{parallelaut} and these definitions in hand, we can now investigate the problem of generating the nonequivalent spanning and near trees of an oriented parallel superedge $G$. Due to their simplicity, we begin with the near trees.

\begin{lemma}
    Let $(G, s, t)$ be an oriented parallel superedge, such that each $G_i$ is non-parallel. Then, $G_a$ and $G_b$ are oriented-isomorphic if and only if $\sigma G_a$ and $\sigma G_b$ are oriented-isomorphic.
\end{lemma}
\begin{proof}
      Let $N = (N_1, N_2, \ldots, N_k)$ and $N^\prime = (N_1^\prime, N_2^\prime, \ldots, N_k^\prime)$ be near trees of $G$, such that $N = \sigma N^\prime$ for some $\sigma \in \text{Aut}_\text{or}(G)$. If $G_a, G_b \in E_i$, then $G_a = \tau G_b$ for some oriented isomorphism $\tau$ of $G_b$.
    As a result,
    \[\sigma G_a = \sigma (\tau G_b) = \sigma \tau \sigma^{-1}(\sigma G_b).\]
    But $\sigma \tau \sigma^{-1}$ is an oriented isomorphism of $\sigma G_b$, since $\sigma,$ $\tau$, and $\sigma^{-1}$ are. Thus, $\sigma G_a$ and $\sigma G_b$ belong to the same oriented isomorphism class. If $G_a$ and $G_b$ do not belong to the same class, then let us assume, for the purpose of contradiction, that $\sigma G_a$ and $\sigma G_b$ do belong to the same class. Then, $\sigma G_a = \tau \sigma G_b$ for some oriented isomorphism $\tau$ of $\sigma G_b$, so left-multiplying by $\sigma^{-1}$ gives us that
    \[G_a = \sigma^{-1}\tau\sigma(G_b).\]
    But $\sigma^{-1}\tau\sigma$ is an oriented isomorphism of $G_b$, so $G_a$ and $G_b$ are oriented-isomorphic, a contradiction. We can thus conclude that $G_a, G_b \in E_i$ if and only if $\sigma G_a, \sigma G_b \in \sigma E_i.$
\end{proof}

\begin{theorem}
    \label{parallelneartrees}
    Let $G$ be an oriented parallel superedge with terminals $s$ and $t$, such that each $G_i$ is non-parallel. Then, the nonequivalent near trees of $G$ can be formed by computing the length-$\vert E_j\vert $ multisets of the nonequivalent near trees of each $R_j$ and assigning the near trees in $R_j$'s multiset to the subgraphs in $E_j$ arbitrarily.
\end{theorem}
\begin{proof}
    First, we observe that an automorphism of $G$ can swap principal subgraphs if and only if they belong to the same class $E_i$. Therefore, if $N_1$ and $N_2$ are near trees of $G$,
    then any sequence of swaps that makes $N_1$ equal to $N_2$ (on top of any automorphisms coming from each $\text{Aut}_\text{or}(G_i)$) must occur between principal subgraphs \emph{within} the same $E_i$. Accounting for these internal swaps, we can compute the nonequivalent near trees over each class $E_i$ separately and then compose these near trees together.

    When we partition the $G_i$'s into their oriented isomorphism classes $E_j$, we save the oriented isomorphism $\tau_i$ that makes each $G_i$ equivalent to its class representative $R_j$. It is not hard to see that the nonequivalent near trees of each $G_i$ are the nonequivalent near trees of $R_j$ multiplied by $\tau_i^{-1}$. Therefore, we can encode each near tree $N$ over $E_j$ using the nonequivalent near trees of $R_j$ that is is formed from.
    Using our earlier tuple notation, we can write the near tree $N$ as follows: \[N = (\tau_1^{-1}N_1, \tau_2^{-1}N_2, \ldots, \tau_\ell^{-1}N_\ell).\]
    Let $N$ and $N^\prime$ (defined similarly) be two near trees of the graph $\bigcup_{G_i \in E_j} G_i$, the union of the principal subgraphs in $E_i$. By Theorem \ref{parallelaut}, $N$ is equivalent to $N^\prime$ if and only if
    \[(\tau_1^{-1}N_1, \tau_2^{-1}N_2, \ldots, \tau_\ell^{-1}N_\ell) = s\big(\sigma_1 ({\tau_1}^\prime)^{-1} N_1, \sigma_2 ({\tau_2}^\prime)^{-1}N_2, \ldots, \sigma_k ({\tau_m}^\prime)^{-1}N_m\big),\] where $s \in \text{swaps}(G)$ and each $\sigma_i \in \text{Aut}_\text{or}(G)$. But since all the $G_i$'s belong to the same oriented isomorphism class, they can be swapped in any way. Therefore, $N$ is equivalent to $N^\prime$ if and only if the \emph{multisets} of their near trees are equal (up to the action of each $\tau_i^{-1}$):
    \[[\tau_1^{-1}N_1, \tau_2^{-1}N_2, \ldots, \tau_m^{-1}N_m] = [\sigma_1 ({\tau_1}^\prime)^{-1} N_1, \sigma_2 ({\tau_2}^\prime)^{-1}N_2, \ldots, \sigma_k ({\tau_m}^\prime)^{-1}N_m].\]
    But this is the same as saying that $N$ and $N^\prime$ are formed using the multiset of trees $[N_1, N_2, \ldots, N_m]$ of $R_j$ and multiplying them by the appropriate isomorphisms $\tau_i^{-1}$.
    Each multiset encodes an \emph{essentially unique} way of assigning near trees to the $G_i$ within $E_j$, up to automorphism.
    
    Hence, we can compute the nonequivalent near trees of $G$ as follows. For each oriented isomorphism class $E_j$, we compute the nonequivalent near trees of $R_j$. Then, we compute all length-$\vert E_j\vert$  multisets of $R_j$'s near trees and assign the near trees in each multiset to the $G_i$'s within $E_j$ arbitrarily, multiplying them by $\tau_i^{-1}$ accordingly. This is the desired result.
\end{proof}

\subsubsection{Spanning Trees on Parallel Superedges}
\label{sec:spanningtreesparallelsuperedges}The task of computing the nonequivalent spanning trees of parallel superedges is similar to that of computing the nonequivalent near trees---but instead of assigning near trees to each $G_i$, we assign near trees to all but \emph{one} of them, to which we assign a spanning tree instead.

In particular, let us assume that the principal subgraph that has a spanning tree belongs to the class $E_a$; without loss of generality, we can also assume that this subgraph is $E_a$'s representative, $R_a$. Every other isomorphism class $E_i$ (where $i \neq a$) consists solely of principal subgraphs with near trees, so we can proceed as before with them: for each $R_i$, we generate the length-$\vert E_i\vert$ multisets of $R_i$'s nonequivalent near trees and assign them to the subgraphs in $E_i$ arbitrarily (since the order does not matter up to the swap automorphisms).

We only have to modify this procedure slightly for the class $E_a$. Up to automorphism, the class $E_a$ is distinguished by having a unique choice of spanning tree on one of its subgraphs and a unique multiset of near trees on the rest (but as described above, it does not matter which subgraph gets what tree). Hence, we first consider every possible way of assigning nonequivalent spanning trees to $R_a$; we then compute the length-$(\vert E_a\vert - 1)$ multisets of $R_a$'s nonequivalent near trees and assign them to the remaining subgraphs in $E_a$ arbitrarily. Repeating this procedure for each choice of the class representative $R_a$---the one possessing the spanning tree---produces the desired output, the nonequivalent spanning trees of $G$ up to automorphism.

The theorem below asserts the correctness of this scheme; but since the proof is very similar to Theorem \ref{parallelneartrees} (with the slight modification described above), we omit it here.
\begin{theorem}
    \label{parallelspanningtrees}
    Let $G$ be an oriented parallel superedge with terminals $s$ and $t$, such that each $G_i$ is non-parallel. Then, the nonequivalent spanning trees of $G$ can be formed by modifying the near tree scheme, as described above.
\end{theorem}
\begin{proof}
    The proof for this result is very similar to Theorem \ref{parallelneartrees}; the only difference is that we must set aside one $R_a$ to be our spanning tree and the remaining $G_j$ are near trees. 
\end{proof}

\subsubsection{Testing Oriented Isomorphism}
\label{isomorphism}
In his 1978 doctoral thesis, Jacobo Valdes Ayesta gives a linear-time algorithm to test whether two series-parallel multidigraphs are isomorphic (by adapting Aho, Hopcroft, and Ullman's classic rooted tree isomorphism algorithm)  \cite{valdes}. In his definition for a series-parallel multidigraph $G$, Valdes Ayesta requires that $G$ has a single source and sink vertex, with indegree and outdegree 0, respectively. As a result, there is only one choice of the source terminal $s$ and the sink terminal $t$ of our multidigraph---the source and sink vertices themselves.

Now, let $(G_i, s, t)$ and $(G_j, s, t)$ be two oriented series-parallel graphs. If we direct the edges of $G_i$ and $G_j$ so that they all flow from $s$ to $t$, then whichever isomorphism makes these digraphs isomorphic (if any) must fix $s$ and $t$. Thus, if their directed counterparts are isomorphic---which we can test using Valdes Ayesta's algorithm---then $G_i \equiv_\text{or} G_j$. Conversely, if $G_i \equiv_\text{or} G_j$ under some oriented isomorphism $\tau : V(G_i) \rightarrow V(G_j)$, then let $u, v \in V(G_i)$. We know that $\{u, v\} \in E(G_i)$ if and only if $\{\tau(u), \tau(v)\} \in E(G_j)$. Since $\tau$ fixes $s$ and $t$, it is not hard to see that $\tau$ preserves the $s \rightarrow t$ direction of our digraphs. The directed edge $(u, v)$ is in the directed version of $G_i$ if and only if $(\tau(u), \tau(v))$ is in the directed $G_j$. Thus, the directed counterparts of $G_i$ and $G_j$ are also isomorphic, so the converse holds.

We conclude that $G_i \equiv_\text{or} G_j$ if and only if the directed counterparts of $G_i$ and $G_j$ are isomorphic. Therefore, we can use Valdes Ayesta's algorithm to test whether two oriented series-parallel graphs $G_i$ and $G_j$ are isomorphic. Combining this observation with Theorems \ref{parallelneartrees} and \ref{parallelspanningtrees}, the following algorithms will generate the spanning/near trees of an oriented parallel superedge up to its automorphisms.
\begin{algorithm}[h]
    \caption{$\texttt{parallel\_spanning\_trees}(G)$}
    \label{alg:parallelspanning}
    \begin{algorithmic}
         \REQUIRE{Parallel superedge $G$, with subgraphs $G_1, G_2, \ldots, G_k$}
         \REQUIRE{Each $G_i$ is non-parallel, for $1 \leq i \leq k$}\newline
         
         \IF{$G$ is an edge} \RETURN{$G$} \hfill\COMMENT{Base case; return null graph on two vertices.}
         \ENDIF\\
         
         \STATE $\texttt{noneq\_trees} \gets \prod_{i = 1}^k \texttt{parallel\_spanning\_trees}(G_i)$ \hfill\COMMENT{Compose all choices of nonequivalent trees.}
         \RETURN{\texttt{noneq\_trees}}
    \end{algorithmic}
\end{algorithm}

\begin{algorithm}[h]
    \caption{$\texttt{parallel\_both\_trees}(G)$}
    \label{alg:parallelboth}
    \begin{algorithmic}
         \REQUIRE{Parallel superedge $G$, with subgraphs $G_1, G_2, \ldots, G_k$}
         \REQUIRE{Each $G_i$ is non-parallel, for $1 \leq i \leq k$}
         
         \IF{$G$ is an edge} \RETURN{$\overline{G}$}
         \ENDIF\\
        \STATE{$\texttt{oriented\_isomorphism\_classes} \gets \{\}$}
         
         \FOR{$1 \leq i \leq k$}
         \STATE {$\texttt{found\_match} \gets \texttt{False}$}
         \FOR{$E_j$ \textbf{in} \texttt{oriented\_isomorphism\_classes}}
         \IF {$G_i \cong E_j\texttt{[0]}$}
            \STATE $E_j \gets E_j \cup \{G_i\}$
            \STATE {$\texttt{found\_match} \gets \texttt{True}$}
         \ENDIF
         \ENDFOR
         \IF {\textbf{not} $\texttt{found\_match}$}
            \STATE {$\texttt{oriented\_isomorphism\_classes} \gets \texttt{oriented\_isomorphism\_classes} \cup \{\texttt{[}G_i\texttt{]}\}$}
        \ENDIF
         \ENDFOR
    \end{algorithmic}
\end{algorithm}

\subsection{Analysis}
We now analyze Algorithms \ref{alg:serialspanning}, \ref{alg:serialboth}, \ref{alg:parallelspanning}, and \ref{alg:parallelboth}. Let $T_S(G)$ and $T_P(G)$ be the cost of computing the nonequivalent spanning trees of a serial or parallel superedge $G$, respectively (i.e.\ Algorithms \ref{alg:serialspanning} and \ref{alg:parallelspanning}). In addition, let $B_S(G)$ and $B_P(G)$ be the cost of computing both the nonequivalent spanning and near trees of a nontrivial serial or parallel superedge $G$ (i.e.\ Algorithms \ref{alg:serialboth} and \ref{alg:parallelboth}).

Let us first analyze Algorithms \ref{alg:serialboth} and \ref{alg:parallelboth}; Algorithms \ref{alg:serialspanning} and \ref{alg:parallelspanning} follow as a consequence. Algorithm \ref{alg:serialboth} (\texttt{serial\_both\_trees}) starts by computing the spanning trees and near trees of each (parallel) principal subgraph $G_i$, which takes total time $\sum_{i = 1}^k B_P(G_i).$ To build the nonequivalent spanning trees of $G$, we then compute every possible union of the nonequivalent spanning trees of each $G_i$. By Theorem \ref{serialtrees}, the number of such unions is $\vert \textsf{ST}(G)/\text{Aut}_\text{or}(G)\vert$, the number of nonequivalent trees itself. The cost of performing each union is $\Theta(m + n)$, where $n$ is the number of vertices of $G$, and $m$ is the number of edges in each spanning tree; spanning trees have $m = n - 1$ edges, so each union simply takes time $\Theta(n)$. As a result, the total cost of building the nonequivalent spanning trees is $\Theta\big(n\cdot\vert \textsf{ST}(G)/\text{Aut}_\text{or}(G)\vert\big).$ By nearly identical reasoning, the cost of building the nonequivalent near trees is $\Theta\big(n \cdot \vert \textsf{NT}(G)/\text{Aut}_\text{or}(G)\vert\big).$ The cost of returning all the spanning and near trees is also $\Theta\big(n\cdot [\vert \textsf{ST}(G)/\text{Aut}_\text{or}(G)\vert + \vert \textsf{NT}(G)/\text{Aut}_\text{or}(G)\vert]\big)$, so we have that
\begin{equation}
    B_S(G) = \sum_{i = 1}^k B_P(G_i) + \Theta\left(n\left[\left\vert \frac{\textsf{ST}(G)}{\text{Aut}_\text{or}(G)}\right\vert + \left\vert \frac{\textsf{NT}(G)}{\text{Aut}_\text{or}(G)}\right\vert\right]\right).
\end{equation}
In Algorithm \ref{alg:parallelboth} (\texttt{parallel\_both\_trees}), we start by partitioning the $G_i$'s into $\ell$ oriented isomorphism classes. To do so, we iterate through each of the $k$ principal subgraphs $G_1, G_2, \ldots, 
 G_k$ and test whether they are oriented-isomorphic to any of the class representatives we have discovered so far; if not, we create a new oriented isomorphism class with $G_i$ as its representative.
 
In order to test whether two principal subgraphs are oriented-isomorphic, we use Valdes Ayesta's algorithm for isomorphism on directed series-parallel graphs (see Section \ref{isomorphism}). If $G_a$ and $G_b$ are two oriented series-parallel graphs with $n_i$ vertices and $m_i$ edges, then we can test their isomorphism in time $O(m_i + n_i)$ \cite{valdes}. Even if $G_a$ and $G_b$ do not share the same number of vertices or edges, we can clearly still test their isomorphism in time $O(m_i + n_i)$, where either one of $G_a$ or $G_b$ has $n_i$ vertices and $m_i$ edges. If we have $\ell$ oriented-isomorphism classes with representatives $R_1, R_2, \ldots, R_\ell$, then we compare each $G_i$ to at most $\ell$ distinct representatives in the partitioning step. Because we have $k$ principal subgraphs to partition, it is not hard to see that this step takes time $O\big(k\ell(e + v)\big) = O\big(\ell(m + n)\big)$, where $e + v$ is the average value of $m_i + n_i$ across all $k$ principal subgraphs $G_i$.

Since we have $\ell$ classes, the total cost of our calls to $\texttt{serial\_both\_trees}$ is $\sum_{i = 1}^\ell B_S(R_i).$ We then have to build and return each possible spanning and near tree of $G$; by similar reasoning to the serial case, this step takes time $\Theta\big(n\cdot [\vert \textsf{ST}(G)/\text{Aut}_\text{or}(G)\vert + \vert \textsf{NT}(G)/\text{Aut}_\text{or}(G)\vert]\big)$. Therefore,
\[B_P(G) = \sum_{i = 1}^\ell B_S(R_i) + \Theta\left(n\left[\left\vert \frac{\textsf{ST}(G)}{\text{Aut}_\text{or}(G)}\right\vert + \left\vert \frac{\textsf{NT}(G)}{\text{Aut}_\text{or}(G)}\right\vert\right] + \ell(m + n)\right).\]

A series-parallel graph on $n$ vertices has at most $m = 2n - 3$ edges \cite{BODIRSKY20072091}, so $\ell(m + n) = O(\ell n)$. Since we have $\ell$ class representatives $R_i$, there are $\ell$ ways of choosing the principal subgraph that is assigned a near tree. As a result, $\vert \textsf{ST}(G)/\text{Aut}_\text{or}(G)\vert \geq \ell$, so $n \cdot \vert \textsf{ST}(G)/\text{Aut}_\text{or}(G)\vert \geq \ell n$. Hence, $\ell(m + n) = O(\ell n)$ is a lower-order term, so our expression for $B_P(G)$ simplifies to
\begin{equation}
    B_P(G) = \sum_{i = 1}^\ell B_S(R_i) + \Theta\left(n\left[\left\vert \frac{\textsf{ST}(G)}{\text{Aut}_\text{or}(G)}\right\vert + \left\vert \frac{\textsf{NT}(G)}{\text{Aut}_\text{or}(G)}\right\vert\right]\right).
\end{equation}
If $G$ is a single edge, then $B_S(G) = B_P(G) = \Theta(1)$.

Let us now investigate Algorithms \ref{alg:serialspanning} and \ref{alg:parallelspanning}. In Algorithm \ref{alg:serialspanning} (\texttt{serial\_spanning\_trees}), we compute the nonequivalent spanning trees of each principal subgraph $G_i$ and compose them together. Here, we build and return spanning trees but not near trees; this gives us that
\begin{equation}
    T_S(G) = \sum_{i = 1}^k T_P(G_i) + \Theta\left(n\left\vert \frac{\textsf{ST}(G)}{\text{Aut}_\text{or}(G)}\right\vert\right).
\end{equation}
In Algorithm \ref{alg:parallelspanning} (\texttt{parallel\_spanning\_trees}), we partition our principal subgraphs into $\ell$ oriented isomorphism classes and then compute the nonequivalent spanning trees and near trees on each representative $R_i$. By similar reasoning to $B_P(G)$ (but where we only generate the spanning trees), we find that
\begin{equation}
     T_P(G) = \sum_{i = 1}^\ell B_S(R_i) + \Theta\left(n\left\vert \frac{\textsf{ST}(G)}{\text{Aut}_\text{or}(G)}\right\vert\right).
\end{equation}
If $G$ is a single edge, then $T_S(G) = T_P(G) = \Theta(1)$, as before. These four recurrences exactly determine the asymptotic runtime over all input graphs $G$.

For simplicity, let $B(G)$ equal $B_S(G)$ if $G$ is serial and $B_P(G)$ if $G$ is parallel; we define $T(G)$ in a similar way. In Theorem \ref{runtime} below, we will prove that Algorithms \ref{alg:serialspanning}, \ref{alg:serialboth}, \ref{alg:parallelspanning}, and \ref{alg:parallelboth} are optimal if we want to return each spanning tree explicitly.\footnote{Many spanning tree generation algorithms choose to return their trees \emph{implicitly} rather than explicitly. Smith’s algorithm, for instance, produces its spanning trees in a revolving-door Gray code order, where successive trees differ by exactly one edge. As a result, we can return the vertices and edges of the first tree explicitly and then repeatedly return the two edges that we delete and add to go from one tree to the next. This implicitly encodes all the trees, and since we only spend constant time writing the pair of edge updates for each tree, it takes time $\Theta(m + n + \tau(G))$, which is optimal \cite{10.5555/1984890}.} To this end, let us first establish a few helpful lemmas.
\begin{lemma}
    \label{seriallemma1}
    Let $G$ be an oriented serial superedge, such that each $G_i$ is non-serial. Then,
    \[\sum_{i = 1}^k n_i\left\vert \frac{\emph{\textsf{ST}}(G_i)}{\emph{Aut}_\emph{or}(G_i)}\right\vert = O\left(n\left\vert \frac{\emph{\textsf{ST}}(G)}{\emph{Aut}_\emph{or}(G)}\right\vert \right),\]
    where $G_i$ has $n_i$ vertices, for $1 \leq i \leq k$.
\end{lemma}
\begin{proof}
    Let $a$ be the number of \emph{basic} principal subgraphs of $G$, meaning that they only have one nonequivalent spanning tree. Each of these subgraphs has $n_i = 2$ vertices and $\vert \textsf{ST}(G_i)/\text{Aut}_\text{or}(G_i)\vert = 1$. Therefore,
    \[\sum_{i = 1}^k n_i\left\vert \frac{\textsf{ST}(G_i)}{\text{Aut}_\text{or}(G_i)}\right\vert \leq 2a + \sum_{G_i\,\text{nonbasic}} n_i\left\vert \frac{\textsf{ST}(G_i)}{\text{Aut}_\text{or}(G_i)}\right\vert \leq 2a + n\sum_{G_i\,\text{nonbasic}} \left\vert \frac{\textsf{ST}(G_i)}{\text{Aut}_\text{or}(G_i)}\right\vert.\]
    But since $\vert \textsf{ST}(G)/\text{Aut}_\text{or}(G)\vert \geq 2$ for each nonbasic $G_i$, the sum over all nonbasic $\vert \textsf{ST}(G_i)/\text{Aut}_\text{or}(G_i)\vert$ is strictly less than the product. By Theorem \ref{serialtrees}, the nonequivalent spanning trees of $G$ are formed by taking every possible union of the nonequivalent spanning trees on each $G_i$, so $\vert \textsf{ST}(G)/\text{Aut}_\text{or}(G)\vert = \prod_{i = 1}^k \vert \textsf{ST}(G_i)/\text{Aut}_\text{or}(G_i)\vert$. Combining these two facts together,
    \begin{align*}
        \sum_{i = 1}^k n_i\left\vert \frac{\textsf{ST}(G_i)}{\text{Aut}_\text{or}(G_i)}\right\vert < 2a + n\prod_{G_i\,\text{nonbasic}} \left\vert \frac{\textsf{ST}(G_i)}{\text{Aut}_\text{or}(G_i)}\right\vert &\leq 2a\prod_{i = 1}^k \left\vert \frac{\textsf{ST}(G_i)}{\text{Aut}_\text{or}(G_i)}\right\vert + n \prod_{i = 1}^k \left\vert \frac{\textsf{ST}(G_i)}{\text{Aut}_\text{or}(G_i)}\right\vert\\[6pt]
        &\leq 2n\prod_{i = 1}^k \left\vert \frac{\textsf{ST}(G_i)}{\text{Aut}_\text{or}(G_i)}\right\vert + n \prod_{i = 1}^k \left\vert \frac{\textsf{ST}(G_i)}{\text{Aut}_\text{or}(G_i)}\right\vert\\[6pt]
        &= \Theta\left(n\left\vert \frac{\textsf{ST}(G_i)}{\text{Aut}_\text{or}(G_i)}\right\vert\right),
    \end{align*}
    so $\sum_{i = 1}^k n_i\cdot \vert \textsf{ST}(G_i)/\text{Aut}_\text{or}(G_i)\vert = O\big(n\cdot \vert \textsf{ST}(G)/\text{Aut}_\text{or}(G)\vert\big)$, as desired.
\end{proof}
\begin{lemma}
    \label{seriallemma2}
    Let $G$ be an oriented serial superedge, such that each $G_i$ is non-serial. Then,
    \[\sum_{i = 1}^k n_i\left\vert \frac{\emph{\textsf{NT}}(G_i)}{\emph{Aut}_\emph{or}(G_i)}\right\vert = O\left(n\left\vert \frac{\emph{\textsf{NT}}(G)}{\emph{Aut}_\emph{or}(G)}\right\vert \right),\]
    where $G_i$ has $n_i$ vertices, for $1 \leq i \leq k$.
\end{lemma}
\begin{proof}
    By Theorem \ref{serialneartrees}, the nonequivalent near trees of $G$ are formed by taking the union of (1) the nonequivalent near trees of one principal subgraph $G_i$ and (2) the nonequivalent spanning trees of every other principal subgraph $G_j$, across all $k$ choices of $G_i$. As a result, $\vert \textsf{NT}(G)/\text{Aut}_\text{or}(G)\vert = \sum_{i = 1}^k \vert \textsf{NT}(G_j)/\text{Aut}_\text{or}(G_i)\vert \cdot \prod_{j \neq i} \vert \textsf{ST}(G_j)/\text{Aut}_\text{or}(G_j)\vert$. Using this fact, we find that
    \begin{align*}
        \sum_{i = 1}^k n_i\left\vert \frac{{\textsf{NT}}(G_i)}{\text{Aut}_\text{or}(G_i)}\right\vert  &\leq n \sum_{i = 1}^k \left\vert \frac{\textsf{NT}(G_i)}{\text{Aut}_\text{or}(G_i)}\right\vert\\[6pt]
        &\leq n\sum_{i = 1}^k \left\vert \frac{\textsf{NT}(G_i)}{\text{Aut}_\text{or}(G_i)}\right\vert \prod_{j \neq i}\left\vert \frac{\textsf{NT}(G_i)}{\text{Aut}_\text{or}(G_i)}\right\vert = O\left(n\left\vert \frac{\textsf{NT}(G)}{\text{Aut}_\text{or}(G)}\right\vert \right),
    \end{align*}
    which is the desired result.
\end{proof}
\begin{lemma}
    \label{parallellemma}
    Let $G$ be an oriented parallel superedge, such that each $G_i$ is non-parallel. Then,
    \[\sum_{i = 1}^\ell n_i\left[\left\vert \frac{\emph{\textsf{ST}}(R_i)}{\emph{Aut}_\emph{or}(G_i)}\right\vert + \left\vert \frac{\emph{\textsf{NT}}(R_i)}{\emph{Aut}_\emph{or}(G_i)}\right\vert\right] = O\left(n\left\vert \frac{\emph{\textsf{ST}}(G)}{\emph{Aut}_\emph{or}(G)}\right\vert\right),\]
    where the class representative $R_i$ has $n_i$ vertices, for $1 \leq i \leq k$. 
\end{lemma}
\begin{proof}
By Theorem \ref{parallelspanningtrees}, the nonequivalent spanning trees of $G$ are formed by (1) assigning nonequivalent spanning trees to some representative $R_i$, (2) assigning the rest of $E_i$ a length-($\vert E_i\vert - 1$) multiset of $R_i$'s nonequivalent near trees, and (3) assigning each remaining $E_j$ a length-$\vert E_j\vert$ multiset of $R_j$'s nonequivalent near trees. There are (1) $\ell$ ways of picking the representative $R_i$, (2) $\binom{\vert \textsf{ST}(R_i)/\text{Aut}_\text{or}(R_i)\vert + \vert E_i\vert - 2}{\vert E_i\vert - 1}$ multisets we can assign to the rest of $E_i$, and (3) $\prod_{j \neq i} \binom{\vert \textsf{ST}(R_i)/\text{Aut}_\text{or}(R_i)\vert + \vert E_i\vert - 1}{\vert E_i\vert}$ ways to assign multisets to the remaining $\ell - 1$ representatives. Hence,
    \begin{align*}
        \left\vert \frac{\textsf{ST}(G)}{\text{Aut}_\text{or}(G)}\right\vert = \sum_{i = 1}^\ell \left\vert \frac{\textsf{ST}(R_i)}{\text{Aut}_\text{or}(R_i)}\right\vert\binom{\vert \textsf{NT}(R_i)/\text{Aut}_\text{or}(R_i)\vert + \vert E_i\vert - 2}{\vert E_i\vert - 1}\prod_{j \neq i}\binom{\vert \textsf{NT}(R_i)/\text{Aut}_\text{or}(R_i)\vert + \vert E_i\vert - 1}{\vert E_i\vert}.
    \end{align*}
    Based on the form of $\vert \textsf{ST}(G)/\text{Aut}_\text{or}(G)\vert$, we observe that $\vert \textsf{ST}(G)/\text{Aut}_\text{or}(G)\vert \geq \sum_{i = 1}^\ell \vert \textsf{ST}(G_i)/\text{Aut}_\text{or}(G_i)\vert$ and $\vert \textsf{ST}(G)/\text{Aut}_\text{or}(G)\vert \geq \sum_{i = 1}^\ell \vert \textsf{NT}(G_i)/\text{Aut}_\text{or}(G_i)\vert$. Thus,
    \begin{align*}
        \sum_{i = 1}^\ell n_i\left[\left\vert \frac{\textsf{ST}(R_i)}{\text{Aut}_\text{or}(R_i)}\right\vert + \left\vert \frac{\textsf{NT}(R_i)}{\text{Aut}_\text{or}(R_i)}\right\vert\right] &= \sum_{i = 1}^\ell n_i\left\vert \frac{\textsf{ST}(R_i)}{\text{Aut}_\text{or}(R_i)}\right\vert + \sum_{i = 1}^\ell n_i \left\vert \frac{\textsf{NT}(R_i)}{\text{Aut}_\text{or}(R_i)}\right\vert\\[6pt]
        &\leq n\sum_{i = 1}^\ell \left\vert \frac{\textsf{ST}(R_i)}{\text{Aut}_\text{or}(R_i)}\right\vert + n\sum_{i = 1}^\ell  \left\vert \frac{\textsf{NT}(R_i)}{\text{Aut}_\text{or}(R_i)}\right\vert\\[6pt]
        &\leq 2n\left\vert \frac{\textsf{ST}(G)}{\text{Aut}_\text{or}(G)}\right\vert.
    \end{align*}
    We conclude that
    \[\sum_{i = 1}^\ell n_i\left[\left\vert \frac{\textsf{ST}(R_i)}{\text{Aut}_\text{or}(R_i)}\right\vert + \left\vert \frac{\textsf{NT}(R_i)}{\text{Aut}_\text{or}(R_i)}\right\vert\right] = O\left(n\left\vert \frac{\textsf{ST}(G)}{\text{Aut}_\text{or}(G)}\right\vert\right).\]
\end{proof}
\begin{theorem}
    \label{runtime}
    If $G$ is a series-parallel graph with $n$ vertices and $m$ edges, then \[B(G) = \Theta\left(n\left[\left\vert \frac{\emph{\textsf{ST}}(G)}{\emph{Aut}_\emph{or}(G)}\right\vert + \left\vert \frac{\emph{\textsf{NT}}(G)}{\emph{Aut}_\emph{or}(G)}\right\vert\right]\right) \qquad\text{and}\qquad T(G) = \Theta\left(n\left\vert \frac{\emph{\textsf{ST}}(G)}{\emph{Aut}_\emph{or}(G)}\right\vert\right).\]
\end{theorem}
\begin{proof}
    To prove this result, we will apply strong induction over the depth $d$ of our series-parallel decomposition tree.
    \vspace{5pt}
    \begin{enumerate}
        \item \underline{Base Case}:\\[6pt]
        Let us consider the case where $d = 0$. Then, $G$ is a single edge with endpoints $s$ and $t$. In this case, \texttt{serial\_both\_trees} and \texttt{parallel\_both\_trees} return the single nonequivalent spanning tree of $G$ ($G$ itself) and near tree ($\overline{G}$) in constant time $\Theta(1)$. Since $G$ is a single edge, $n = 2$ and $\vert \textsf{ST}(G)/\text{Aut}_\text{or}(G)\vert = \vert \textsf{NT}(G)/\text{Aut}_\text{or}(G)\vert = 1$, so $\Theta\big(n\cdot [\vert \textsf{ST}(G)/\text{Aut}_\text{or}(G)\vert + \vert \textsf{NT}(G)/\text{Aut}_\text{or}(G)\vert] \big)$ also equals $\Theta(1)$. Hence,
        \[B(G) = \Theta\left(n\left[\left\vert \frac{\text{\textsf{ST}}(G)}{\text{Aut}_\text{or}(G)}\right\vert + \left\vert \frac{\text{\textsf{NT}}(G)}{\text{Aut}_\text{or}(G)}\right\vert\right]\right).\]
        By similar reasoning, $T(G) = \Theta\big(n\cdot \vert \textsf{ST}(G) / \text{Aut}_\text{or}(G)\vert\big)$. Thus, the claim holds when $d = 0$.\\
        \item \underline{Inductive Hypothesis}:\\[6pt]
        Let us assume that the claim is true for all $d \leq h$, for some $h \in \mathbb{N}$. That is, for all series-parallel graphs $G$ with depth  at most $h$, we have that $B(G) = \Theta\big(n\cdot [\vert \textsf{ST}(G)/\text{Aut}_\text{or}(G)\vert + \vert \textsf{NT}(G)/\text{Aut}_\text{or}(G)\vert]\big)$ and $T(G) = \Theta\big(n\cdot \vert \textsf{ST}(G)/\text{Aut}_\text{or}(G)\vert\big)$. We will show that the claim is then true for depth $d = h + 1$.\\

        \item \underline{Induction Step}:\\[6pt]
        Let $G$ be any series-parallel graph with depth $d = h + 1$. Then, the child subtrees of $G$ (which are the decomposition trees of each $G_i$) each have depth at most $h$. If $G$ is serial, then we have that
        \begin{align*}
            B_S(G) &= \sum_{i = 1}^k B(G_i) + \Theta\left(n\left[\left\vert \frac{\textsf{ST}(G)}{\text{Aut}_\text{or}(G)}\right\vert + \left\vert \frac{\textsf{NT}(G)}{\text{Aut}_\text{or}(G)}\right\vert\right]\right)\\[6pt]
            &= \Theta\left(n\left[\left\vert \frac{\textsf{ST}(G)}{\text{Aut}_\text{or}(G)}\right\vert + \left\vert \frac{\textsf{NT}(G)}{\text{Aut}_\text{or}(G)}\right\vert\right]\right) + \sum_{i = 1}^k \Theta\left(n_i\left[\left\vert \frac{\textsf{ST}(G_i)}{\text{Aut}_\text{or}(G_i)}\right\vert + \left\vert \frac{\textsf{NT}(G_i)}{\text{Aut}_\text{or}(G_i)}\right\vert\right]\right)\\[6pt]
            &= \Theta\left(n\left[\left\vert \frac{\textsf{ST}(G)}{\text{Aut}_\text{or}(G)}\right\vert + \left\vert \frac{\textsf{NT}(G)}{\text{Aut}_\text{or}(G)}\right\vert\right] + \sum_{i = 1}^k n_i\left[\left\vert \frac{\textsf{ST}(G_i)}{\text{Aut}_\text{or}(G_i)}\right\vert + \left\vert \frac{\textsf{NT}(G_i)}{\text{Aut}_\text{or}(G_i)}\right\vert\right]\right),
        \end{align*}
        where each $G_i$ has $n_i$ vertices. By Lemmas \ref{seriallemma1} and \ref{seriallemma2}, $\sum_{i = 1}^k n_i\cdot [\vert \textsf{ST}(G_i)/\text{Aut}_\text{or}(G_i)\vert] = O\big(n\cdot\vert \textsf{ST}(G)/\text{Aut}_\text{or}(G)\vert\big)$ and $\sum_{i = 1}^k n_i\cdot [\vert \textsf{NT}(G_i)/\text{Aut}_\text{or}(G_i)\vert] = O\big(n\cdot\vert \textsf{NT}(G)/\text{Aut}_\text{or}(G)\vert\big)$. Therefore, the summation term is at most the same order as the leading term, so we deduce that
        \[B_S(G) = \Theta\left(n\left\vert \frac{\textsf{ST}(G_i)}{\text{Aut}_\text{or}(G_i)}\right\vert\right).\]
        If $G$ is parallel, then we proceed similarly:
        \begin{align*}
            B_P(G) &= \sum_{i = 1}^\ell B(R_i) + \Theta\left(n\left[\left\vert \frac{\textsf{ST}(G)}{\text{Aut}_\text{or}(G)}\right\vert + \left\vert \frac{\textsf{NT}(G)}{\text{Aut}_\text{or}(G)}\right\vert\right]\right)\\[6pt]
            &= \Theta\left(n\left[\left\vert \frac{\textsf{ST}(G)}{\text{Aut}_\text{or}(G)}\right\vert + \left\vert \frac{\textsf{NT}(G)}{\text{Aut}_\text{or}(G)}\right\vert\right]\right) + \sum_{i = 1}^\ell \Theta\left(n_i\left[\left\vert \frac{\textsf{ST}(R_i)}{\text{Aut}_\text{or}(R_i)}\right\vert + \left\vert \frac{\textsf{NT}(R_i)}{\text{Aut}_\text{or}(R_i)}\right\vert\right]\right)\\[6pt]
            &= \Theta\left(n\left[\left\vert \frac{\textsf{ST}(G)}{\text{Aut}_\text{or}(G)}\right\vert + \left\vert \frac{\textsf{NT}(G)}{\text{Aut}_\text{or}(G)}\right\vert\right] + \sum_{i = 1}^\ell n_i\left[\left\vert \frac{\textsf{ST}(R_i)}{\text{Aut}_\text{or}(R_i)}\right\vert + \left\vert \frac{\textsf{NT}(R_i)}{\text{Aut}_\text{or}(R_i)}\right\vert\right]\right),
        \end{align*}
        By Lemma \ref{parallellemma}, the summation term has order at most $O\big(n\cdot\vert \textsf{ST}(G)/\text{Aut}_\text{or}(G)\vert\big)$, so
        \[B_P(G) =  \Theta\left(n\left[\left\vert \frac{\textsf{ST}(G)}{\text{Aut}_\text{or}(G)}\right\vert + \left\vert \frac{\textsf{NT}(G)}{\text{Aut}_\text{or}(G)}\right\vert\right] \right).\] Combining our asymptotic expressions for $B_S(G)$ and $B_P(G)$ gives us that
        \[B(G) = \Theta\left(n\left[\left\vert \frac{\textsf{ST}(G)}{\text{Aut}_\text{or}(G)}\right\vert + \left\vert \frac{\textsf{NT}(G)}{\text{Aut}_\text{or}(G)}\right\vert\right] \right),\] as desired.

        Let us now consider $T(G)$. If $G$ is serial, then we have that
        \begin{align*}
            T_S(G) = \sum_{i = 1}^k T(G_i) + \Theta\left(n\left\vert \frac{\textsf{ST}(G)}{\text{Aut}_\text{or}(G)}\right\vert\right)  &= \Theta\left(n\left\vert \frac{\textsf{ST}(G)}{\text{Aut}_\text{or}(G)}\right\vert\right) + \sum_{i = 1}^\ell \Theta\left(n_i\left\vert \frac{\textsf{ST}(R_i)}{\text{Aut}_\text{or}(R_i)}\right\vert\right)\\[6pt]
             &= \Theta\left(n\left\vert \frac{\textsf{ST}(G)}{\text{Aut}_\text{or}(G)}\right\vert + \sum_{i = 1}^\ell n_i\left\vert \frac{\textsf{ST}(R_i)}{\text{Aut}_\text{or}(R_i)}\right\vert\right).
        \end{align*}
        By Lemma \ref{parallellemma}, the summation term once again has order at most $O\big(n\cdot\vert\textsf{ST}(G)/\text{Aut}_\text{or}(G)\vert\big)$, so as we might expect,
        \[T_S(G) = \Theta\left(n\left\vert \frac{\textsf{ST}(G)}{\text{Aut}_\text{or}(G)}\right\vert\right).\]
        Finally, if $G$ is parallel,
        then
         \begin{align*}
            T_P(G) &= \sum_{i = 1}^\ell B_S(R_i) + \Theta\left(n\left\vert \frac{\textsf{ST}(G)}{\text{Aut}_\text{or}(G)}\right\vert\right).\\[6pt]
             &= \Theta\left(n\left\vert \frac{\textsf{ST}(G)}{\text{Aut}_\text{or}(G)}\right\vert\right) + \sum_{i = 1}^\ell \Theta\left(n_i\left[\left\vert \frac{\textsf{ST}(R_i)}{\text{Aut}_\text{or} (R_i)}\right\vert + \left\vert \frac{\textsf{NT}(R_i)}{\text{Aut}_\text{or} (R_i)}\right\vert\right]\right)\\[6pt]
             &= \Theta\left(n\left\vert \frac{\textsf{ST}(G)}{\text{Aut}_\text{or}(G)}\right\vert + \sum_{i = 1}^\ell n_i\left[\left\vert \frac{\textsf{ST}(R_i)}{\text{Aut}_\text{or} (R_i)}\right\vert + \left\vert \frac{\textsf{NT}(R_i)}{\text{Aut}_\text{or} (R_i)}\right\vert\right]\right).
        \end{align*}
        By Lemma \ref{parallellemma}, the summation term has order at most $O\big(n\cdot \vert \textsf{ST}(G)/\text{Aut}_\text{or}(G)\vert\big)$, so $T_P(G) = \Theta\big(n\cdot \vert \textsf{ST}(G)/\text{Aut}_\text{or}(G)\vert\big)$. Using our asymptotic expressions for $T_S(G)$ and $T_P(G)$, we conclude that
        \[T(G) = \Theta\left(n\left\vert \frac{\textsf{ST}(G)}{\text{Aut}_\text{or}(G)}\right\vert\right),\]
        which is the desired result.
    \end{enumerate}
    By induction, we conclude that \[B(G) = \Theta\left(n\left[\left\vert \frac{\text{\textsf{ST}}(G)}{\text{Aut}_\text{or}(G)}\right\vert + \left\vert \frac{\text{\textsf{NT}}(G)}{\text{Aut}_\text{or}(G)}\right\vert\right]\right) \qquad\text{and}\qquad T(G) = \Theta\left(n\left\vert \frac{\textsf{ST}(G)}{\text{Aut}_\text{or}(G)}\right\vert\right).\]
\end{proof}
By Theorem \ref{runtime}, Algorithms \ref{alg:serialspanning} and \ref{alg:parallelspanning} run with output-linear time $O(n \cdot \vert \textsf{ST}(G)/\text{Aut}_\text{or}(G)\vert\big)$, which is optimal.

\section{Semioriented Graphs}
\label{sec:semioriented}
For the most part, the problem of generating the nonequivalent spanning trees of a semioriented series-parallel graph reduces to the oriented case. However, since semioriented graphs have terminals but neither distinguished sources nor sinks, we now have to consider automorphisms that \emph{exchange} the terminals $s$ and $t$. Let $G$ be a semioriented series-parallel graph with terminals $\{s, t\}$, and let $\text{Aut}_\text{semi}(G)$ denote the subgroup of $G$'s automorphisms that either fix or exchange $s$ and $t$:
\begin{definition}
Given a semioriented series-parallel graph $(G, \{s, t\})$, we write that $\emph{Aut}_\emph{semi}(G, \linebreak \{s, t\}) = \big\{\sigma \in \emph{Aut}(G) : \{s, t\} = \{\sigma(s), \sigma(t)\}\big\}$. For brevity, we simply write $\emph{Aut}_\emph{semi}(G)$.
\end{definition}
The definition for $\text{Aut}_\text{or}(G)$ on semioriented graphs is the same as before. Naturally, we define our equivalence relation on $\textsf{ST}(G)$ as follows.
\begin{definition}
    Let $T_1$ and $T_2$ be spanning trees of the semioriented graph $G$. Then, $T_1$ and $T_2$ are equivalent if and only if $T_1 = \sigma T_2$ for some $\sigma \in \emph{Aut}_\emph{semi}(G)$.
\end{definition}
For the remainder of this section, it will be helpful to think of the different principal subgraphs as having an ``orientation.\rq\rq\ We define a ``forward\rq\rq\ and ``backward\rq\rq\ version of each of our subgraphs, as follows:
\begin{definition}
    For each $G_i$, where $1 \leq i \leq k$, let $\overrightarrow{G_i}$ be the oriented series-parallel graph $(G_i, s, t)$ and $\overleftarrow{G_i}$ be $(G_i, t, s)$.
\end{definition}
The theorems below prove helpful properties about the structure of $\text{Aut}_\text{semi}(G)$ for semioriented serial and parallel superedges, by letting us rewrite $\text{Aut}_\text{semi}(G)$ in terms of $\text{Aut}_\text{or}(G)$.
\begin{theorem}
    \label{semiserialaut}
    Let $G$ be a semioriented serial superedge with terminals $s = s_1, s_2, \ldots, s_{k + 1} = t$, such that each $G_i$ is non-serial. Then, if $\overrightarrow{G_i} = r_i\overleftarrow{G_{k - i + 1}}$ (where $r_i : V(G_i) \rightarrow V(G_{k - i + 1})$) for $1 \leq i \leq k$, then $\emph{Aut}_\emph{semi}(G) = \{\sigma, r\sigma : \sigma \in \emph{Aut}_\emph{or}(G)\}$, where $r = r_1r_2\cdots r_k.$ Otherwise, $\emph{Aut}_\emph{semi}(G) = \emph{Aut}_\emph{or}(G)$.
\end{theorem}
\begin{proof}
    Let us assume that $\overrightarrow{G_i} = r_i\overleftarrow{G_{k - i + 1}}$ for $1 \leq i \leq k,$ and let $\tau \in \text{Aut}_\text{semi}(G)$ such that $\tau(s) = t$ and $\tau(t) = s$. Then, the product $r = r_1r_2\cdots r_k \in \text{Aut}_\text{semi}(G)$ and has order 2, so $\tau = r^2\tau = r(r\tau).$ But $r\tau \in \text{Aut}_\text{or}(G)$, since $(r\tau)(s) = r(\tau(s)) =  r(t) = s$ and $(r\tau)(t) = r(\tau(t)) =r(s) = t.$ As a result, if $\tau$ exchanges $s$ and $t$, then $\tau = r\sigma$ for some $\sigma \in \text{Aut}_\text{or}(G)$.
    
    If $\tau$ fixes $s$ and $t$, on the other hand, then clearly $\tau \in \text{Aut}_\text{or}(G)$. We deduce that $\text{Aut}_\text{semi}(G) \subseteq \{\sigma, r\sigma : \sigma \in \text{Aut}_\text{or}(G)\}$. However, each function $\sigma, r\sigma \in \text{Aut}_\text{or}(G)$ is also in $\text{Aut}_\text{semi}(G)$. It follows that $\text{Aut}_\text{semi}(G) = \{\sigma, r\sigma : \sigma \in \text{Aut}_\text{or}(G)\}.$
    
    The remainder of the proof, where we show that $\text{Aut}_\text{semi}(G) = \text{Aut}_\text{or}(G)$ if no such $r$ exists, borrows elements from Theorem \ref{serialaut}. In particular, let us assume that there exists no function $r = r_1r_2\cdots r_k$ such that each $\overrightarrow{G_i} = r_i\overleftarrow{G_{k - i + 1}}$.
    
    The terminals $s_1, s_2, \ldots, s_{k + 1}$ are the unique vertices of $G$ that are included in every path $P$ from $s_1$ to $s_{k + 1}$. Moreover, $P$ always visits these vertices in this order: $s_1, s_2, \ldots, s_{k + 1}$. Now, let $\sigma \in \text{Aut}_\text{semi}(G)$. Due to the bijective properties of automorphisms, $\sigma(s_1), \sigma(s_2), \ldots, \sigma(s_{k + 1})$ are the unique vertices of $G$ included in every path $P^\prime$ from $\sigma(s_1)$ to $\sigma(s_{k + 1})$, and $P^\prime$ visits these vertices in this order. Since $\sigma$ preserves the path order of our terminal vertices, $\sigma(s_i)$ must be the $i$th terminal along the path from $\sigma(s_1)$ to $\sigma(s_{k + 1})$.
    
    If $\sigma(s) = t$ and $\sigma(t) = s$, then $\sigma(s_i) = s_{k - i + 1}$. It is not hard to see, then, that $\sigma G_i = G_{k - i + 1}$. Therefore, $\sigma$ induces an isomorphism from $G_i$ to $G_{k - i + 1}$, so $\sigma = \sigma_1\sigma_2\ldots\sigma_k$ where $\sigma_i : V(G_i) \rightarrow V(G_{k - i + 1})$. But there is no function $r = r_1r_2\ldots r_k$ such that $\overrightarrow{G_i} = r_i\overleftarrow{G_{k - i + 1}}$, so no such $\sigma$ can exist. As a result, $\sigma(s) = s$ and $\sigma(t) = t$.
    
    As a result, every $\sigma \in \text{Aut}_\text{semi}(G)$ has $\sigma(s) = s$ and $\sigma(t) = t$, so $\text{Aut}_\text{semi}(G) \subseteq \text{Aut}_\text{or}(G).$ But $\text{Aut}_\text{or}(G)$ is clearly a subgroup of $\text{Aut}_\text{semi}(G)$, so $\text{Aut}_\text{semi}(G) = \text{Aut}_\text{or}(G)$, as desired.
\end{proof}
\begin{theorem}
    \label{semiparallelaut}
    Let $G$ be a semioriented parallel superedge with terminals $s$ and $t$, such that each $G_i$ is non-parallel. Then, if there exists a permutation $\tau \in S_k$ such that $\overrightarrow{G_i} = r_i\overleftarrow{G_{\tau(i)}}$ (where $r_i : V(G_i) \rightarrow V\big(G_{\tau(i)}\big)$) for $1 \leq i \leq k$, then $\emph{Aut}_\emph{semi}(G) = \{\sigma, r\sigma : \sigma \in \emph{Aut}_\emph{or}(G)\}$, where $r = r_1r_2\cdots r_k.$ Otherwise, $\emph{Aut}_\emph{semi}(G) = \emph{Aut}_\emph{or}(G)$.
\end{theorem}
\begin{proof}
    Let us assume that there exists a permutation $\tau \in S_k$ such that $\overrightarrow{G_i} = r_i\overleftarrow{G_{\tau(i)}}$ for $1 \leq i \leq k$. If $\rho \in \text{Aut}_\text{semi}(G)$ and $\rho$ fixes $s$ and $t$, then clearly $\rho \in \text{Aut}_\text{or}(G)$. If, on the other hand, $\rho$ exchanges $s$ and $t$, then by similar reasoning to Theorem \ref{serialaut}, $\rho G_i = G_{j(i)}$ where $1 \leq j(i) \leq k$. But since each $G_i$ is sent to a unique $G_{j(i)}$ under $\rho$, we realize that $j(i)$ permutes the indices $i = 1, 2 \ldots, k$. That is, $j \in S_k$.
    
    Since $\rho$ exchanges $s$ and $t$, we also have that $\rho \overrightarrow{G_i} = \overleftarrow{G_{j(i)}}$. But since $\overrightarrow{G_i}$ equals $r_i\overleftarrow{G_{\tau(i)}}$, we find that $\overleftarrow{G_{j(i)}} = \rho r_i\overleftarrow{G_{\tau(i)}}$. It is not hard to see, then, that $\sigma_i = \rho r_i \in \text{Aut}_\text{or}(G)$ for $1 \leq i \leq k$. Therefore, $\overleftarrow{G_{j(i)}} = \sigma_i\overleftarrow{G_{\tau(i)}}$, so substituting on the left-hand side gives us that $\rho \overrightarrow{G_i} = \sigma_i \overleftarrow{G_{\tau(i)}}$. But $\overleftarrow{G_{\tau(i)}} = r_i\overrightarrow{G_i}$, so $\rho \overrightarrow{G_i} = (\sigma_i r_i)\overleftarrow{G_{\tau(i)}}$, and \[\rho \overrightarrow{G_i} = (r_i\sigma_i)\overleftarrow{G_{\tau(i)}}\] for $1 \leq i \leq k$. Each of the $r_i$'s and $\sigma_i$'s are disjoint, so we can write that $\rho = r\sigma$, where $r = r_1r_2\cdots r_k$ and $\sigma = \sigma_1\sigma_2\cdots \sigma_k$. Hence, if $\rho$ exchanges $s$ and $t$, then $\rho = r\sigma$ for some $\sigma \in \text{Aut}_\text{semi}(G)$.
    
    Thus, we find that every automorphism $\rho \in \text{Aut}_\text{semi}(G)$ is of the form $\sigma$ or $r\sigma$, where $\sigma \in \text{Aut}_\text{or}(G)$, so 
 $\text{Aut}_\text{semi}(G) \subseteq \{\sigma, r\sigma : \sigma \in \text{Aut}_\text{or}(G)\}$. It is straightforward to verify the reverse containment, that $\{\sigma, r\sigma : \sigma \in \text{Aut}_\text{or}(G)\} \subseteq \text{Aut}_\text{semi}(G)$. We conclude that
 \[\text{Aut}_\text{semi}(G) = \{\sigma, r\sigma : \sigma \in \text{Aut}_\text{or}(G)\}.\]
 It is a matter of formality to show that $\text{Aut}_\text{semi}(G) = \text{Aut}_\text{or}(G)$ if no such permutation $\tau \in S_k$ exists. 
\end{proof}

\subsection{Serial Lexicographical Constraint}
\label{sec:seriallexicographicalconstraint}
Using Theorem \ref{semiserialaut}, let us now investigate the problem of generating the nonequivalent spanning trees of a serial superedge $G$. If $\text{Aut}_\text{semi}(G) = \text{Aut}_\text{or}(G)$, then it suffices to use Algorithm \ref{alg:serialspanning} for oriented serial superedges. On the other hand, if $\text{Aut}_\text{semi}(G)$ allows us to exchange $s$ and $t$, then we can make optimizations based on the symmetry of $G_i$ and $G_{k - i + 1}$. In particular, let us assume that the principal subgraph $(G_i, s_i, t_i)$ of $G$ has the nonequivalent spanning trees $\{T_{i,1}, T_{i,2}, … , T_{i,j(i)}\}$ up to $\text{Aut}_\text{or}(G)$, for $1 \leq i \leq \lceil k/2\rceil$. Then, its ``mirror\rq\rq\  subgraph $(G_{k - i + 1}, s_{k - 1 + 1}, t_{k - i + 1})$ has the nonequivalent trees $\{r_iT_{i, 1}, r_iT_{i, 2}, \ldots, r_iT_{i, j(i)}\}$. Therefore, we only need to compute the nonequivalent trees once for $G_i$ and apply the isomorphism $r_i$ to find them for $G_{k - i + 1}$, instead of computing them separately for both subgraphs. For simplicity, let us assume that we index the nonequivalent spanning trees of $G_i$ and $G_{k - i + 1}$ the same way, so that $T_{k - i + 1, x} = r_iT_{i, x}$ for $1 \leq x \leq j(i)$.

If we simply compute the union of the nonequivalent spanning trees on each $(G_i, s_i, t_i)$---like Algorithm \ref{alg:serialspanning}---then we will produce several spanning trees of $G$ which equivalent under $\text{Aut}_\text{semi}(G)$. In particular, the tree \[T = \big(T_{1,a}, T_{2,b}, T_{3, c},\ldots, T_{k -2, c^\prime}, {T_{k - 1, b^\prime}}, T_{k, a^\prime}\big)\] is equivalent to the tree $T^\prime = (T_{1, a^\prime}, T_{2, b^\prime}, T_{3, c^\prime},\ldots, T_{k - 2, c}, {T_{k - 1, b}}, T_{k,a})$ where we have reversed the order of the tree indices that we assign to each $G_i$ (i.e.\ tree $abc$ is equivalent to tree $cba$). As a result, we produce equivalent trees in pairs. To avoid this problem, we can generate our spanning trees under a lexicographical constraint: we require that the sequence of trees assigned to the left half of $G$---$T_{1, a}, T_{2, b}, \ldots, T_{\lfloor k/2\rfloor, c}$---is lexicographically greater than or equal to the sequence on the right half starting from the end: $T_{k, a^\prime}, T_{k - 1, b^\prime}, \ldots, T_{\lceil k/2\rceil, c^\prime}$. That is, we require that the sequence $a, b, \ldots, c$ is lexicographically greater than or equal to $a^\prime, b^\prime, \ldots, c^\prime$.

For example, if our tree is $\big(T_{1, 1}, T_{2, 3}, T_{3, 7}, T_{4, 4}\big)$, formed by picking tree 1 on $G_1$, tree 2 on $G_2$, tree 7 on $G_3$, and tree 4 on $G_4$, then its partner tree is $\big(T_{1, 4}, T_{2, 7}, T_{3, 3}, T_{4, 1}\big)$, where we pick tree 4 on $G_1$, tree 7 on $G_2$, tree 3 on $G_3$, and tree 1 on $G_4$. The first tree has tree sequence $(1, 3, 7, 4)$, so the reversed right half $(4, 7)$ is lexicographically higher than the left half $(1, 3)$. In the partner tree, though, the left half is $(4, 7)$, which is lexicographically higher than the reversed right half $(1, 3)$. As a result, we only generate the first partner, and we never generate a tree alongside its counterpart.

The following theorem formally proves the correctness of this lexicographical scheme.

\begin{theorem}
    \label{seriallex}
    Let $G$ be a semioriented serial superedge with terminals $s$ and $t$, such that each $G_i$ is non-parallel. Then, the nonequivalent spanning trees of $G$ can be formed by computing every union of the nonequivalent oriented spanning trees of each $(G_i, s_i, t_i)$ under our serial lexicographical constraint.
\end{theorem}
\begin{proof}
    Let each oriented principal subgraph $(G_i, s_i, t_i)$ of $G$ have the nonequivalent spanning trees $\{T_{i, 1}, T_{i, 2}, \ldots, T_{i,j(i)}\}$ under some arbitrary indexing of its $m_i$ spanning trees. By Theorem \ref{serialtrees}, we can form the set $\mathcal{O}$ of $G$'s nonequivalent \emph{oriented} spanning trees by computing every union \[\big(T_{1, a_1}, T_{2, a_2}, \ldots , T_{k, a_k}\big)\] of nonequivalent trees on each $(G_i, s_i, t_i)$, where $1 \leq a_i \leq m_i$ for each $i$.  Now, let $A = \big(T_{1, a_1}, T_{2, a_2}, \ldots, \linebreak T_{k, a_k}\big)$ and $B = \big(T_{1, b_1}, T_{2, b_2}, \ldots m T_{k, b_k}\big)$ be any two trees in $\mathcal{O}$ such that $A = \sigma B$ for some $\sigma \in \text{Aut}_\text{semi}(G)$. That is,\[\big(T_{1, a_1}, T_{2, a_2}, \ldots T_{k, a_k}\big) = \sigma\big(T_{1, b_1}, T_{2, b_2}, \ldots T_{k, b_k}\big).\] If $\sigma \in \text{Aut}_\text{or}(G)$, then $A$ and $B$ are oriented-equivalent. But by Theorem \ref{serialtrees}, $\mathcal{O}$ contains the oriented-nonequivalent spanning trees of $G$ without repeats. The trees $A$ and $B$ must therefore represent the same spanning tree in $\mathcal{O}$; that is, $a_i = b_i$ for $1 \leq i \leq k$.
    
    On the other hand, suppose that $\sigma \notin \text{Aut}_\text{or}(G)$. Now, let us assume that there is a third (but, as we shall see, not necessarily distinct) tree $C$ in $\mathcal{O}$ such that $A = \tau C$ for some $\tau \in \text{Aut}_\text{semi}(G)$. If $\tau \in \text{Aut}_\text{or}(G)$, then $A$ must equal $C$ by our prior reasoning. But if $\tau \in \text{Aut}_\text{semi}(G)$, then it is not hard to see that $B$ and $C$ are oriented-equivalent, so $B$ equals $C$ by the same token. Therefore, $B$ is the \emph{only} other tree in $\mathcal{O}$ that $A$ can be equivalent to.
    
    By Theorem \ref{semiserialaut}, $\sigma$ must equal $r\sigma^\prime$, where (1) $\sigma^\prime \in \text{Aut}_\text{or}(G)$, and (2) $\overrightarrow{G_i} = r_i\overleftarrow{G_{k - i + 1}}$ for $1 \leq i \leq k$, and (3) $r = r_1r_2 \ldots r_k$. As stated in Section \ref{sec:seriallexicographicalconstraint}, we can assume that the nonequivalent oriented spanning trees of $(G_i, s_i, t_i)$ and $(G_{k - i + 1}, s_{k - i + 1}, t_{k - i + 1})$ are indexed the same way, so $T_{k - i + 1, x} = r_iT_{i, x}$ for $1 \leq x \leq j(i)$. Since $A = \sigma B$, we have that $A = r\sigma^\prime B$, so $B = r({\sigma^\prime})^{-1}A$. But $r$ reverses the order of the $G_i$'s and $({\sigma^\prime})^{-1} \in \text{Aut}_\text{or}(G)$, so $B$ must be oriented-equivalent to the tree $\big(T_{1, a_k}, T_{2, a_{k - 1}}, \ldots, T_{k, a_1}\big)$. But this latter tree belongs to $\mathcal{O}$, which does not repeat any oriented-equivalent trees; thus, $B = \big(T_{1, a_k}, T_{2, a_{k - 1}}, \ldots, T_{k, a_1}\big)$.
    
    Clearly, then, only one of $A$ and $B$ can be included under our lexicographical constraint; if they both satisfy the left half being greater than or equal to the right half reversed, then $A = B$. Hence, our lexicographical constraint does not generate any equivalent (i.e.\ duplicate) spanning trees of $G$ up to $\text{Aut}_\text{semi}(G)$.
    
    However, every tree $T \in \textsf{ST}(G)$ is equivalent to some tree $T^\prime$ in $\mathcal{O}$. If $\text{Aut}_\text{semi}(G) = \text{Aut}_\text{or}(G)$, then $T^\prime$ is included up to the lexicographical constraint; if not, then either $T^\prime$ or $\sigma T^\prime$ (where $\sigma \in \text{Aut}_\text{semi}(G) - \text{Aut}_\text{or}(G)$) is included up to the constraint. As a result, every tree $T \in \textsf{ST}(G)$ is equivalent to some tree in the oriented union under our lexicographical constraint. This scheme does not repeat any equivalent trees, so it correctly generates all the nonequivalent spanning trees of $G$ up to $\text{Aut}_\text{semi}(G)$, as desired.
\end{proof}
\subsection{Parallel Lexicographical Constraint}
\label{sec:parallellexicographicalconstraint}
Once again, parallel superedges are a slightly more complex variation on serial superedges. To see why, let $T$ be a spanning tree of $G$. From Section \ref{sec:spanningtreesparallelsuperedges} (and Theorem \ref{parallelspanningtrees}), every oriented isomorphism class $E_i$ is either assigned (1) a length-$\vert E_i\vert$ multiset of $R_i$'s near trees or (2) a spanning tree on $R_i$ and a length-$\big(\vert E_i\vert - 1\big)$ multiset of $R_i$'s near trees on its remaining subgraphs. As a result, each class $E_i$ has a collection $\mathcal{M}$ of multisets that it could be assigned: let us index these multisets so that the spanning tree multisets (those where $R_i$ is assigned a spanning tree) are all numbered higher than the near tree multisets. As we know, every spanning tree on $G$ is formed by picking a multiset on each class $E_i$, but where precisely one class is assigned a spanning tree multiset.

By Theorem \ref{parallelspanningtrees}, we can form the nonequivalent oriented spanning trees of $G$ by simply taking every possible union of multisets, where exactly one of the multisets in our union has a spanning tree. Up to $\text{Aut}_\text{or}(G)$, we can encode each spanning trees of $G$ with the tuple of the multisets it is assigned; that is, \[T \equiv \big(M_{1, m(1)}, M_{2, m(2)}, \ldots, M_{\ell, m(\ell)}\big),\] where we assign each class $E_i$ its multiset with index $m(i)$.

As before, let us partition our subgraphs into oriented isomorphism classes $E_i$. But now, if we ever have that $R_i = \sigma R_j$ for some $\sigma \in \text{Aut}_\text{semi}(G) - \text{Aut}_\text{or}(G)$ and two class representatives $R_i$ and $R_j$, we say that $E_j = \overleftarrow{E_i}$ (or equivalently, $E_i = \overleftarrow{E_j}$): that is, $E_j$ is the \emph{reverse} of $E_i$. For clarity, we also write that $E_i = \overrightarrow{E_i}$. Note that $E_i$ and $E_j$ are not necessarily distinct---$\overrightarrow{E_i}$ could equal $\overleftarrow{E_i}$---since there might exist a class representative $R_i$ which is horizontally symmetric.

When we allow semioriented automorphisms that reverse the orientation of our subgraphs, the oriented scheme produces equivalent trees of $G$ in pairs. In particular, let us assume that the multisets of $\overrightarrow{E_i}$ and $\overleftarrow{E_i}$ are indexed the same way; then, the tree where we assign multiset $a_i$ to each $\overrightarrow{E_i}$ and $a_i^\prime$ to each $\overleftarrow{E_i}$ is the same as the tree where $\overrightarrow{E_i}$ gets multiset $a_i^\prime$ and $\overleftarrow{E_i}$ gets $a_i$.

To avoid generating duplicates, we once again impose a lexicographical constraint. 
We partition our classes into a set $A$ of forward subgraphs and a set $B$ of backward subgraphs such that for $1 \leq i \leq \ell$, either $\overrightarrow{E_i} \in A$ and $\overleftarrow{E_i} \in B$ or vice versa; we only allow $\overrightarrow{E_i}$ to be in both $A$ and $B$ if $\overrightarrow{E_i} = \overleftarrow{E_i}$. Our constraint requires that the sequence of multisets in $A$ is lexicographically greater than or equal to the sequence in $B$. As we prove in Theorem \ref{parallellex} below, this lexicographical scheme will correctly skip duplicates for reasons similar to the serial case.
\begin{theorem}
\label{parallellex}
Let $G$ be a semioriented parallel superedge with terminals $s$ and $t$, such that each $G_i$ is non-serial. Then, the nonequivalent spanning trees of $G$ can be formed by generating the nonequivalent oriented spanning trees of $(G, s, t)$ under our parallel lexicographical constraint.
\end{theorem}
\begin{proof}
Let us assume that we have partitioned our oriented isomorphism classes into sets $A$ and $B$ as described above, so that $A = \big\{\overrightarrow{E_{a(1)}}, \overrightarrow{E_{a(2)}}, \ldots, \overrightarrow{E_{a(p)}}\big\}$ and $B = \big\{\overleftarrow{E_{a(1)}}, \overleftarrow{E_{a(2)}}, \ldots, \overleftarrow{E_{a(p)}}\big\}$ for some $p \leq \ell$. In addition, let $\mathcal{O}$ be the set of $G$'s nonequivalent spanning trees up to $\text{Aut}_\text{or}(G)$, computed using the procedure in Section \ref{sec:spanningtreesparallelsuperedges}, and let $\mathcal{S}$ be the trees in $\mathcal{O}$ that pass our lexicographical constraint. We will show that $\mathcal{S}$ is the set of $G$'s nonequivalent spanning trees up to $\text{Aut}_\text{semi}(G)$.

Let $T$ and $T^\prime$ be any two trees in $\mathcal{O}$ such that $T = \sigma T^\prime$ for some $\sigma \in \text{Aut}_\text{semi}(G)$. If $\sigma \in \text{Aut}_\text{or}(G)$, then $T$ clearly equals $T^\prime$, since the trees in $\mathcal{O}$ are nonequivalent up to $\text{Aut}_\text{or}(G)$ by Theorem \ref{serialtrees}. On the other hand, if $\sigma \notin \text{Aut}_\text{semi}(G)$, then we let $T_A$ and $T_B$ be the sequence of $T$'s multisets on $A$ and $B$, respectively:
\[T_A = \big(M_{a(1), b(1)}, M_{a(2), b(2)}, \ldots, M_{a(p), b(p)}\big)\quad \text{and}\quad T_B = r\big(M_{a(1), c(1)}, M_{a(2), c(2)}, \ldots, M_{a(p), c(p)}\big),\]
where class $\overrightarrow{E_i}$ gets multiset $M_{a(i), b(i)}$ and class $\overleftarrow{E_i}$ gets $r_iM_{a(i), b(i)}$. (As per Theorem \ref{semiparallelaut}, $r_i$ is an isomorphism from $\overrightarrow{E_i}$ to $\overleftarrow{E_i}$ that fixes $s$ and $t$, and $r = r_1r_2\cdots r_k$; this effectively lets us index the multisets of $\overrightarrow{E_i}$ and $\overleftarrow{E_i}$ the same way.) We similarly define $T^\prime_A$ and $T^\prime_B$ for tree $T^\prime$:
\[T^\prime_A = \big(M_{a^\prime(1), b^\prime(1)}, M_{a^\prime(2), b^\prime(2)}, \ldots, M_{a^\prime(p), b^\prime(p)}\big)\quad \text{and}\quad T^\prime_B = r\big(M_{a^\prime(1), c^\prime(1)}, M_{a^\prime(2), c^\prime(2)}, \ldots, M_{a^\prime(p), c^\prime(p)}\big),\]
where class $\overrightarrow{E_i}$ gets multiset $M_{a^\prime(i), b^\prime(i)}$ and class $\overleftarrow{E_i}$ gets $r_iM_{a^\prime(i), b^\prime(i)}$. Since $\sigma \notin \text{Aut}_\text{or}(G)$, by Theorem \ref{semiparallelaut}, we have that $\sigma = r\tau$ for some $\tau \in \text{Aut}_\text{or}(G)$. In addition, $r$ swaps the multisets on $A$ and $B$, and $\tau$ does not change the multiset assignments (since $\tau$ is an oriented automorphism), so $T_A = rT_B^\prime$ and $T_B = rT_A^\prime$. Therefore, 
\begin{align*}
    \big(M_{a^\prime(1), b^\prime(1)}, M_{a^\prime(2), b^\prime(2)}, \ldots, M_{a^\prime(p), b^\prime(p)}\big) &= \big(M_{a^\prime(1), c^\prime(1)}, M_{a^\prime(2), c^\prime(2)}, \ldots, M_{a^\prime(p), c^\prime(p)}\big)\\[6pt]
    \big(M_{a(1), c(1)}, M_{a(2), c(2)}, \ldots, M_{a(p), c(p)}\big) &= \big(M_{a^\prime(1), b^\prime(1)}, M_{a^\prime(2), b^\prime(2)}, \ldots, M_{a^\prime(p), b^\prime(p)}\big),
\end{align*} so $b(i) = c^\prime(i)$ and $c(i) = b^\prime(i)$ for $1 \leq i \leq p$. Thus, $T$ has the same sequence of multisets on $A$ that $T^\prime$ has on $B$, and vice versa. If we have the even stronger condition that $T$ and $T^\prime$ have the same multisets on $A$ and the same multisets on $B$, then $T$ and $T^\prime$ are oriented-equivalent by Theorem \ref{parallelspanningtrees} (they have the same multisets on each class $E_i$, even without orientation). If $T$ and $T^\prime$ are not oriented-equivalent, then either $T_A > T_B$ and so $T^\prime_A < T^\prime_B$ lexicographically, or vice versa (we rule out $T_A$ and $T_B$ being lexicographically equal since this implies that $T$ is oriented-equivalent to $T^\prime$). As a result, if $T = \sigma T^\prime$ for some $\sigma \in \text{Aut}_\text{semi}(G)$, then $\mathcal{S}$ contains exactly one of $T$ and $T^\prime$.

And so, our lexicographical scheme does not generate any equivalent trees up to $\text{Aut}_\text{semi}(G)$. Now, we must show that it generates all of the nonequivalent spanning trees of $G$. Every tree $T \in \textsf{ST}(G)$ is oriented-equivalent to some tree $T^\prime \in \mathcal{O}$. If $T^\prime$ satisfies the lexicographical constraint, we are done. If not, then there must exist a tree $\sigma T^\prime \in \mathcal{O}$ which does, for some $\sigma \in \text{Aut}_\text{semi}(G) - \text{Aut}_\text{or}(G)$. The tree $\sigma T^\prime$ is equivalent to $T^\prime$  up to $\text{Aut}_\text{semi}(G)$ and therefore equivalent to $T$ as well. Hence, $T$ is equivalent to either $T^\prime$ or $\sigma T^\prime$, one of which must be in $\mathcal{S}$. We conclude that $\mathcal{S}$ contains the nonequivalent trees of $G$ up to $\text{Aut}_\text{semi}(G)$, as desired.

If there is no way to partition our oriented isomorphism classes into sets $A$ and $B$, then there exists no permutation $\tau \in S_k$ such that $\overrightarrow{G_i} = r_i\overleftarrow{G_{\tau(i)}}$. By Theorem \ref{semiparallelaut}, $\text{Aut}_\text{semi}(G)$ simply equals $\text{Aut}_\text{or}(G)$, so $\mathcal{O}$ itself contains the nonequivalent spanning trees of $G$.
\end{proof}
\subsection{Algorithms for Semioriented Graphs} It is quite straightforward to adapt the existing algorithms for oriented series-parallel graphs for the semioriented case. In fact, the lexicographical constraint is only a ``top-level\rq\rq\ modification to the algorithms for oriented series-parallel graphs. When we compute the nonequivalent spanning trees of semioriented graphs, our recursive calls all operate on \emph{oriented} principal subgraphs $(G_i, s_i, t_i)$: we only apply the constraint at the topmost layer of our algorithm to avoid generating global, semioriented-equivalent trees of $G$. Intuitively, this is because the automorphisms that exchange $s$ and $t$ apply to $G$ as a whole, not each $G_i$ locally.

To generate the nonequivalent spanning trees of semioriented serial superedges, we first check whether $\text{Aut}_\text{semi}(G) = \text{Aut}_\text{or}(G)$ by testing if $(G_i, s_i, t_i)$ is oriented-isomorphic to $(G_{k - i + 1}, t_{k - i + 1}, s_{k - i + 1})$ for $1 \leq i \leq k$ (see Section \ref{isomorphism}). If this is the case, we compute the nonequivalent spanning trees of each oriented $(G_i, s_i, t_i)$ (recalling that the trees for $(G_i, s_i, t_i)$ are the same as $(G_{k - i + 1}, s_{k - i + 1}, t_{k - i + 1})$ once we apply the correct isomorphism $r_i$ from $G_i$ to $G_{k - i + 1}$). By Theorem \ref{seriallex}, the nonequivalent spanning trees of $G$ can be formed by computing the union of the nonequivalent spanning trees of each $(G_i, s_i, t_i)$ under our lexicographical constraint---which we have just computed. Thus, we can generate the nonequivalent spanning trees of $G$ by correctly assigning spanning trees to each $G_i$, making sure that the sequence of trees on the left half is lexicographically greater than or equal to the sequence on the right half reversed (see Section \ref{sec:seriallexicographicalconstraint}).

We make similar modifications for semioriented parallel superedges. When we partition our subgraphs into oriented isomorphism classes and check whether $\overrightarrow{G_i}$ is oriented-isomorphic to the class representative $\overrightarrow{R_j}$, we now also test whether $\overrightarrow{G_i}$ is oriented-isomorphic to $\overleftarrow{R_j}$. This allows us to easily divide our classes into a forward direction $A$ (of $\overrightarrow{E_i}$'s) and a backward direction $B$ (of corresponding $\overleftarrow{E_i}$'s). By Theorem \ref{parallellex}, we can form the nonequivalent spanning trees of $(G, \{s, t\})$ by assigning spanning tree and near tree multisets to each $E_i$, making sure to enforce our parallel lexicographical constraint (see Section \ref{sec:parallellexicographicalconstraint}).

Since these two semioriented algorithms are simple modifications of Algorithms \ref{alg:serialspanning} and \ref{alg:parallelspanning}, we omit their pseudocode here. Like their oriented counterparts, the semioriented algorithms run in worst-case output-linear time, in this case $\Theta\big(\vert \textsf{ST}(G)/\text{Aut}_\text{semi}(G)\vert\big)$.

\section{Unoriented Graphs}
\label{unoriented}
Unoriented series-parallel graphs prove to be the most challenging case, since there are no explicitly designated terminals $s$ and $t$ and therefore no restrictions on how automorphisms act on these terminals. Consequently, unoriented graphs can possess many more symmetries beyond the ones we have thus considered: local automorphisms of each principal subgraph $G_i$, swaps between oriented-equivalent $G_i$, and global reversals that exchange $s$ and $t$.

For example, let $G$ be the cycle graph $C_n$, which is an unoriented series-parallel graph since we can pick any two vertices as our terminals---forming a parallel superedge with two principal subgraphs in the process. We still have all of the original automorphisms contained within $\text{Aut}_\text{semi}(G)$, across all choices of the terminals $s$ and $t$. But once we fix these terminals, $G$ has automorphisms that cannot be explained through swaps, reversals, and local automorphisms alone. For example, we can \emph{rotate} $G$, and most of $G$'s rotations do not fit into our original framework, as shown in Figure \ref{fig:emergentauts} below. In fact, $\text{Aut}(G) \cong D_{2n}$, the dihedral group of order $2n$, since we can make $n$ rotations and also reflect $G$ over one of its $n$ axes.
\begin{figure}[h]
\label{fig:emergentauts}
\centering
\begin{tikzpicture}
    \node[circle,thick,draw,minimum size=0.7cm,fill=gray!20] (s) at (0, 0) {$s$};
    \node[circle,thick,draw,minimum size=0.7cm] (2) at (1.5, 1.5) {2};
    \node[circle,thick,draw,minimum size=0.7cm] (3) at (4.5, 1.5) {3};
    \node[circle,thick,draw,minimum size=0.7cm,fill=gray!20] (t) at (6, 0) {$t$};
    \node[circle,thick,draw,minimum size=0.7cm] (6) at (1.5, -1.5) {6};
    \node[circle,thick,draw,minimum size=0.7cm] (5) at (4.5, -1.5) {5};
    \draw[-, ultra thick, black] (2) -- (s);
    \draw[-, ultra thick, gray!75] (s) -- (6);
        \draw[-, ultra thick, gray!75] (5) -- (t);
    \draw[-, ultra thick ] (t) -- (3)  node[midway, above right = 0.5mm]{\huge $\circlearrowleft$};
    \draw[-, ultra thick] (2) -- (3) node[midway, below = 0.5mm] {$G_1$};
    \draw[-, ultra thick, gray!75] (5) -- (6) node[midway, above = 0.5mm] {$G_2$};

    \node[circle,thick,draw,minimum size=0.7cm] (21) at (9, 0) {2};
    \node[circle,thick,draw,minimum size=0.7cm] (31) at (10.5, 1.5) {3};
    \node[circle,thick,draw,minimum size=0.7cm,fill=gray!20] (t1) at (13.5, 1.5) {$t$};
    \node[circle,thick,draw,minimum size=0.7cm] (51) at (15, 0) {5};
    \node[circle,thick,draw,minimum size=0.7cm,fill=gray!20] (s1) at (10.5, -1.5) {$s$};
    \node[circle,thick,draw,minimum size=0.7cm] (61) at (13.5, -1.5) {6};
    \draw[-, ultra thick, black] (21) -- (s1);
    \draw[-, ultra thick, gray!75] (s1) -- (61);
        \draw[-, ultra thick, gray!75] (51) -- (t1);
    \draw[-, ultra thick ] (t1) -- (31);
    \draw[-, ultra thick] (21) -- (31) node[midway, below right=0.5mm] {$G_1$};
    \draw[-, ultra thick, gray!75] (51) -- (61) node[midway, above left=0.5mm] {$G_2$};

     \node (GEquals) at ($(t.center)!0.5!(21.center)$) {\Large $=$};
\end{tikzpicture}

\caption{When $G$ is unoriented and we pick terminals $s$ and $t$, $\emph{Aut}(G)$ often contains automorphisms beyond swaps, reversals, and the local automorphisms of each $G_i$. Here, $G = C_6$ and has two oriented-equivalent principal subgraphs $G_1$ and $G_2$. If we rotate $G$ counterclockwise by one edge, then we preserve the original structure of $G$ without swapping $G_1$ and $G_2$, exchanging $s$ and $t$, or applying nontrivial automorphisms from $\emph{Aut}_\emph{or}(G_1)$ and $\emph{Aut}_\emph{or}(G_2)$.}
  \label{fig:rotate}

\end{figure}
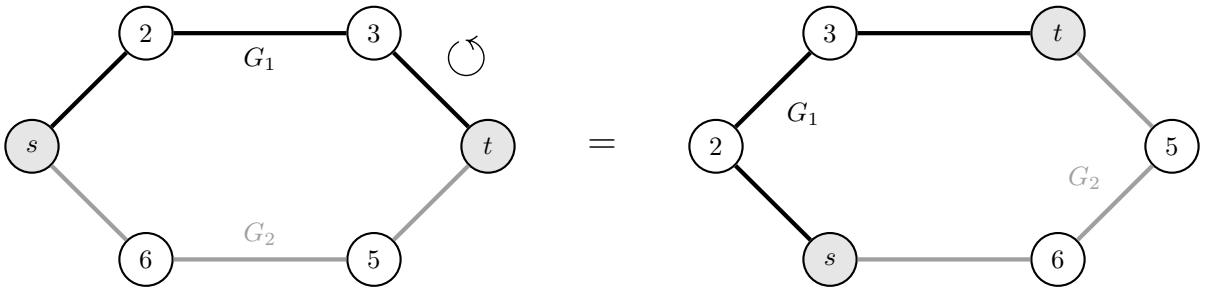

However, we can still impose some structure on an unoriented serial superedge $G$. No matter how we choose our outer terminals $s = s_1$ and $t = s_{k + 1}$, the inner terminals $s_2, s_3, \ldots, s_k$ are precisely the articulation points of $G$. As a result, the inner terminals are determined even if the outer terminals $s$ and $t$ are not, so we can still speak meaningfully about $G$'s inner principal subgraphs $G_2, G_3, \ldots, G_{k - 1}$. Regardless of the specific choice of $s$ and $t$, we also observe that $s_2$ must be the sink terminal of $G_1$ and $s_k$ the source terminal of $G_k$. Therefore, the unoriented graph $G$ by itself gives us many of the properties of oriented and semioriented series-parallel graphs; we can thus recurse as before on the inner subgraphs $G_2, G_3, \ldots, G_{k - 1}$, finding their oriented-nonequivalent spanning trees. The only open question is how we find the nonequivalent spanning trees of $G$'s ``half-oriented\rq\rq\ subgraphs $G_1$ and $G_k$.

Based on our cycle graph example above, unoriented parallel superedges at first appear to be the more complex of the two cases. However, in many parallel superedges, we can find a set of two vertices $\{u, v\}$ that is invariant under automorphism, meaning that every automorphism $\sigma \in \text{Aut}(G)$ either fixes or exchanges $u$ and $v$. When this is the case, we set $s = u$ and $t = v$, and $\text{Aut}(G)$ is simply equal to $\text{Aut}_\text{semi}(G, \{s, t\})$. Consequently, the nonequivalent spanning trees of the unoriented graph $G$ are the same as the nonequivalent spanning trees of the semioriented graph $(G, \{s, t\})$.

For instance, let us take the same unoriented parallel superedge $G$ in Figure \ref{fig:rotate} but add an edge between $s$ and $t$, as shown in \ref{fig:addanedge} below. Then, every automorphism $\sigma \in \text{Aut}(G)$ must either fix or exchange $s$ and $t$ since they are the only vertices of degree 3, so $\text{Aut}(G) = \text{Aut}_\text{semi}(G, \{s, t\})$.
\begin{figure}[h]
\label{fig:addanedge}
\centering
\begin{tikzpicture}
    \node[circle,thick,draw,minimum size=0.7cm,fill=gray!20] (s) at (0, 0) {$s$};
    \node[circle,thick,draw,minimum size=0.7cm] (2) at (1.5, 1.5) {2};
    \node[circle,thick,draw,minimum size=0.7cm] (3) at (4.5, 1.5) {3};
    \node[circle,thick,draw,minimum size=0.7cm,fill=gray!20] (t) at (6, 0) {$t$};
    \node[circle,thick,draw,minimum size=0.7cm] (6) at (1.5, -1.5) {6};
    \node[circle,thick,draw,minimum size=0.7cm] (5) at (4.5, -1.5) {5};
    \draw[-, ultra thick, black] (2) -- (s);
    \draw[-, ultra thick] (s) -- (6);
        \draw[-, ultra thick] (5) -- (t);
    \draw[-, ultra thick ] (t) -- (3);
    \draw[-, ultra thick] (2) -- (3);
    \draw[-, ultra thick] (5) -- (6);
    \draw[-, ultra thick] (s) -- (t);
\end{tikzpicture}

\caption{When we add an edge connecting $s$ and $t$, every automorphism $\sigma \in \emph{Aut}(G)$ must now fix or exchange $s$ and $t$. Hence, $\emph{Aut}(G)$ is equal to $\emph{Aut}_\emph{semi}(G, \{s, t\})$, so we can treat the unoriented graph $G$ as the semioriented graph $(G, \{s, t\})$ instead.}
\end{figure}
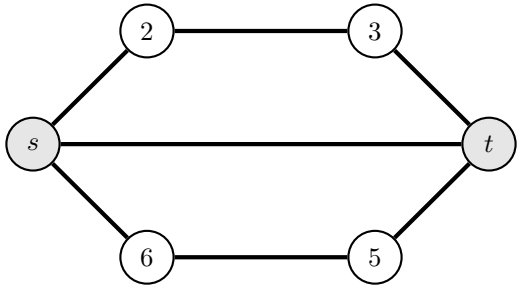

As a result, unoriented parallel superedges often reduce to the semioriented case, provided that we can find two vertices $u$ and $v$ with this property. Based on these observations, we will present a few open questions about unoriented graphs in the next section.

\section{Conclusion}
\label{sec:conclusion}
In this paper, we have shown how to optimally generate the spanning trees of a series-parallel graph $G$---under its two most standard definitions---up to $G$'s automorphisms. After introducing the broader problem of generating spanning trees up to automorphism, we showed how to produce the nonequivalent spanning trees of an oriented series-parallel graph $G$ in output-linear time. We then demonstrated how we can perform lexicographical bookkeeping to adapt these oriented algorithms for semioriented series-parallel graphs instead. Finally, we introduced unoriented series-parallel graphs, which sometimes---but not always---evade the constraints of oriented and semioriented graphs.

To close, we offer the following open questions about generating spanning trees up to automorphism:\\

\begin{enumerate}
    \item $\text{[Unoriented graphs.]}$ Is there an algorithm that generates the nonequivalent spanning trees of an unoriented series-parallel graph $G$ in output-linear time $\Theta\big(n\cdot \vert \textsf{ST}(G)/\text{Aut}(G)\vert\big)$? Can we reduce unoriented graphs to the oriented/semioriented case? If so, do we still need to compute $\text{Aut}(G)$, or the vertex orbits of $G$ up to $\text{Aut}(G)$ (so that we can find two vertices that have to be mapped onto each other and thus act as terminals)?\\
    
    \item $\text{[Implicit generation.]}$ Throughout our discussion, we attempt to generate the nonequivalent spanning trees of a series-parallel graph $G$ explicitly rather than implicitly. Is it possible to generate the nonequivalent spanning trees implicitly instead, so that we can return them in time essentially proportional to $\vert \textsf{ST}(G)/\text{Aut}(G)\vert$, the number of nonequivalent trees itself? Can we implicitly generate these nonequivalent spanning trees under a Gray code order, as per Knuth and Smith? \cite{10.5555/1984890, smith}\\

    \item $\text{[General graphs.]}$ Can we extend the ideas here to generate the nonequivalent spanning trees of a general graph $G$? From the standpoint of complexity theory, how hard is this problem on general graphs? Can we significantly improve upon brute force, and if so, do we still have to compute $\text{Aut}(G)$ explicitly?\\

    \item $\text{[Counting nonequivalent trees.]}$ When is it substantially easier to \emph{count} the nonequivalent spanning trees of a graph $G$  than generate them? Are there cases where generating the trees is effectively our only option to count them all?\\
\end{enumerate}

In the problem of generating spanning trees up to graph automorphism, series-parallel graphs prove to be an especially interesting case. For many graphs, we can find closed-form formulas for $\vert \textsf{ST}(G)/\text{Aut}(G)\vert$ by using tools from abstract algebra, like Burnside's lemma; however, it is another problem entirely to generate these trees, not just count them. Series-parallel graphs are an example where both problems are conceptually nontrivial yet computationally easy: they give us a means to simultaneously count and generate these trees, even though there is not one singular closed-form formula or set of nonequivalent trees that applies to all series-parallel graphs of a given size. 

\section*{Acknowledgements}
We would like to extend our deep thanks to Peter Kagey of California State Polytechnic University, Pomona, as well as Dagan Karp, Michael Orrison, Nicholas Pippenger, and Timothy Randolph of Harvey Mudd College, for many illuminating conversations about enumerative combinatorics, abstract algebra, and graph algorithms.

\bibliographystyle{siamplain}
\bibliography{references}
\end{document}